\documentclass[letterpaper,twocolumn,10pt]{article} %
\usepackage{usenix-2020-09}

\usepackage[title]{appendix}
\usepackage{subcaption}
\usepackage[utf8]{inputenc}
\usepackage[greek,english]{babel}
\usepackage{alphabeta}
\usepackage{booktabs}  %
\usepackage{dcolumn}   %
\usepackage{multirow}  %
\usepackage{rotating}
\usepackage{xspace}
\usepackage{hyphenat}  %

\usepackage{color}
\definecolor{Bittersweet}{RGB}{169,120,81}
\definecolor{OliveGreen}{RGB}{44,94,26}
\definecolor{RoyalPurple}{RGB}{108,73,152}

\usepackage{enumerate}
\usepackage{balance}
\usepackage{fancyvrb}
\fvset{%
fontsize=\footnotesize,
numbers=left,
xleftmargin=12pt
}
\usepackage{textcomp}
\usepackage{amsmath}
\usepackage{mathrsfs}
\usepackage{amsfonts}
\usepackage{wasysym}
\usepackage{tabularx}
\usepackage{pifont}
\usepackage{microtype} %
\usepackage{mathtools} %
\usepackage{mathrsfs} %

\usepackage{url}

\usepackage{longfbox}
\makeatletter
\newdimen\@tempdimd
\makeatother
\fboxset{rounded,outer-width=\linewidth, border-color=black, padding-top=3pt, padding-bottom=3pt}

\usepackage{ulem}
\usepackage{array}    %

\usepackage{dblfloatfix}

\usepackage[noeka]{mathrmletter}
\usepackage{arydshln}
\usepackage[english]{babel}
\usepackage{blindtext}

\newif\ifextended
\newif\ifminted
\newif\ifintrouble
\newif\ifcuttext

\ifx\buildextended\undefined
\else
    \extendedtrue
\fi
\ifintrouble
\cuttexttrue
\else
\cuttextfalse
\fi

\ifx\buildminted\undefined
\else
  \mintedtrue
\fi
\ifminted
  \usepackage{minted}
\else
\usepackage{colortbl} %
\fi

\newcommand{\textred}[1]{\begingroup \color{red} #1\endgroup}
\ifx\noeditingmarks\undefined
   \newcommand{\pgwrapper}[2]{\textred{#1: #2}}
   \newcommand{\pgwrapperb}[1]{\textbf{#1}}
\else
   \newcommand{\pgwrapperb}[1]{}
   \newcommand{\pgwrapper}[2]{}
\fi

\ifx\noeditingmarks\undefined

\else

\fi

\newcommand{\sys}{Pomp\={e}-SRO\xspace}

\makeatletter
\renewcommand*{\@fnsymbol}[1]{\ensuremath{\ifcase#1\or *\or \star\or \dagger\or \ddagger\or
   \mathsection\or \mathparagraph\or \|\or **\or \dagger\dagger
   \or \ddagger\ddagger \else\@ctrerr\fi}}
\makeatother

\newcommand{\MPworker}[2]{{\color{#1}\vrule\vrule}{\marginpar{\color{#1} #2}}}
\ifx\noeditingmarks\undefined
    \newcommand{\SP}[1]{\MPworker{red}{STVS: #1}}

    \newcommand{\JLP}[1]{\MPworker{blue}{JL: #1}}
    \newcommand{\MPtg}[1]{\MPworker{green}{TG: #1}}
    \newcommand{\SAP}[1]{\MPworker{magenta}{SGA: #1}}
\else
    \newcommand{\SP}[1]{}
    \newcommand{\JLP}[1]{}
    \newcommand{\MPtg}[1]{}
    \newcommand{\SAP}[1]{}
\fi
\newcommand\rmv[1]{}

\newcommand{\techReportOnly}[1]{}

\newcommand{\R}{\mathcal{R}}

\newcommand{\C}{\mathcal{C}}

\newcommand{\Po}{\mathcal{P}}
\newcommand{\Io}{\mathcal{I}}

\def\compactify{\itemsep0in \topsep2pt \parsep=0.00in \partopsep=0pt
\leftmargin2em}
\let\latexusecounter=\usecounter

  {\begin{list}{\labelitemi}{\itemsep0in \topsep2pt \parsep0.00in
  \partopsep=0pt \leftmargin\parindent}}%
  {\end{list}}

  {\begin{list}{\labelitemi}{\itemsep6pt \topsep6pt \parsep0.00in
  \partopsep=0pt \leftmargin\parindent}}%
  {\end{list}}

  {\begin{list}{\labelitemi}{\itemsep3pt \topsep3pt \parsep0.00in
  \partopsep=0pt \leftmargin\parindent}}%
  {\end{list}}
  {\begin{list}{\labelitemi}{\itemsep0in \topsep2pt \parsep0.00in
  \partopsep=0pt \leftmargin1em}}%
  {\end{list}}

\usepackage{amsthm}
\theoremstyle{definition}

\theoremstyle{definition}
\newtheorem{definition}{Definition}[section]

\theoremstyle{plain}
\newtheorem{lemma}{Lemma}[section]
\newtheorem{theorem}{Theorem}[section]
\newtheorem{assumption}{Assumption}[section]

\def\discretionaryslash{\discretionary{/}{}{/}}
{\catcode`\/\active
\gdef\URLprepare{\catcode`\/\active\let/\discretionaryslash
        \def~{\char`\~}}}%
\def\URL{\bgroup\URLprepare\realURL}%
\def\realURL#1{\tt #1\egroup}%

\renewcommand\paragraph[1]{\vspace{3px} \noindent \textbf{#1}}

\begin{document} 

\title{\Large \bf Equal Opportunity: A Correctness Condition for Ordered Consensus\vspace{-1em}}

\author{
{\rm Yunhao Zhang}\thanks{The first two authors contributed equally to this work.}\\
Cornell University
\and
{\rm Haobin Ni}\footnotemark[1]\\
University of Washington
\and
{\rm Soumya Basu}\\
OpenReserve Holdings
\and
{\rm Shir Cohen}\\
Cornell University
\and
{\rm Maofan Yin}\\
University of California, Santa Barbara
\and
{\rm Lorenzo Alvisi}\\
Cornell University
\and
{\rm Robbert van Renesse}\\
Cornell University
\and
{\rm Qi Chen}\\
Microsoft Research
\and
{\rm Lidong Zhou}\\
Microsoft Research
}

\date{}
\maketitle
\pagestyle{plain}

\begin{abstract}
In proof-of-stake blockchains based on State Machine Replication
(SMR), the order of transactions directly affects client-visible
financial outcomes. \emph{Ordered consensus} augments the SMR
specification by imposing correctness conditions on transaction
ordering, with a focus on limiting Byzantine influence. However,
real-world ordering attacks can occur even when these conditions
hold, often enabled by advantages such as faster networks or proximity
to the blockchain infrastructure that allow an adversary to
systematically bias outcomes without violating the conditions.
To address this gap, we extend ordered consensus with
a new model and correctness condition based on {\em equal opportunity},
a notion of fairness widely used in legal contexts.
Equal opportunity requires that candidates who are equally
qualified---according to criteria deemed {\em relevant}---have
equal chances of being selected (here, for a given
position in the total order).  We show how carefully introduced
randomness can bound ordering bias, and we introduce the Secret Random
Oracle (SRO), a fault-tolerant abstraction for generating such
randomness. We present two SRO constructions, based on trusted
hardware and threshold verifiable random functions, respectively, and
use them to build \sys, a new ordered consensus protocol that
mitigates well-known ordering attacks.
Our evaluation shows that {\sys} effectively mitigates front-running
and sandwich attacks at a moderate cost to latency.

\end{abstract}

\section{Introduction}
\label{s:intro}

This paper extends {\em ordered consensus}~\cite{zhang2020byzantine}
by motivating, expressing, and enforcing \textit{equal opportunity}, a
correctness condition governing how a state machine replication
(SMR)~\cite{schneider90implementing} protocol orders client requests.

At its core, SMR requires a set of replicas to agree on a single,
totally ordered sequence of client requests ({\em i.e.,} a ledger).
When SMR is used for fault tolerance, this order serves only to ensure
that all correct replicas process requests consistently.
As long as all requests from correct clients eventually appear in the
ledger, their relative order does not affect correctness.
In blockchains, however, the specific order matters, as it can
directly impact the financial rewards associated with the transactions
recorded in the ledger.

{\em Ordered consensus}~\cite{zhang2020byzantine} formalizes
correctness conditions on this order.  Each replica associates an {\em
  ordering indicator} ({\em e.g.,} a timestamp) with every client
request, expressing how it prefers to order requests relative to one
another.  Using these indicators, one can show that while Byzantine
influence over the ledger order cannot be eliminated, it can be
curtailed. In particular, ordered consensus can guarantee {\em
  ordering linearizability}: if the lowest indicator assigned by any
correct node for command \(c_{2}\) exceeds the highest indicator
assigned to \(c_{1}\), then \(c_{1}\) appears before \(c_{2}\) in the
ledger. This guarantee is enforced by the Pomp\={e}
protocol~\cite{zhang2020byzantine}.

The starting point of this paper is the observation that limiting
Byzantine influence alone is not sufficient: ordering linearizability
fails to capture fairness requirements arising in real-world
ordering-based attacks with significant financial
consequences~\cite{qin2022quantifying}.
Ordering linearizability only constrains the order of requests whose
timestamps are unambiguously separated; it says nothing about how
to order requests that arrive close together in time, leaving room
for network advantages to systematically determine outcomes.

Consider, for example, the unfair practice known in financial markets
as {\em front-running}, in which a party learns of a large pending buy
order for a certain stock and submits its own buy order in advance, so
it can later profit from the price increase driven by the ensuing
large buy order.  Such behavior has been widely documented in both
traditional financial markets~\cite{lewis14flashboys} and
decentralized ones~\cite{daian20flash,torres2021frontrunner}, often
enabled by advantages such as lower network latency.  While financial
regulations prohibit these practices, Pomp\={e} does not prevent them:
requests from clients with faster network access are likely to appear
earlier in the output ledger.

Our first contribution is a model of fairness grounded in real-world
concerns.  Drawing on the legal and economic literature on {\em equal
  opportunity}, we observe that, whenever items are ranked, their
positions depend on their characteristics (or {\em features}).  Some
features are deemed {\em relevant} to the ranking's purpose,
while others are {\em irrelevant}.  For example, U.S. employers and
lending agencies are legally forbidden from considering certain
irrelevant features (so-called {\em protected classes}) when making
decisions~\cite{civilrightact,originalempolymentact}.  Intuitively, a
fair ranking depends only on relevant features and ignores all others:
entries with indistinguishable relevant features should have equal
chances of being ranked ahead of one another.

Building on the expressiveness of ordered consensus, we instantiate
this notion of equal opportunity for  blockchains.  In
this context, two features are typically considered relevant: the {\em
  time} at which clients submit requests and the {\em fee} they attach
to them.  Other features, such as geographic location, are
irrelevant.  Existing ordered consensus protocols, however, do not
distinguish between relevant and irrelevant features, nor do they
reason about this distinction.  As a result, protocols such as
Pomp\={e}~\cite{zhang2020byzantine}, Aequitas~\cite{kelkar20order},
and Themis~\cite{kelkar2021themis} remain vulnerable to adversaries
that exploit irrelevant features---for example, faster network
connectivity---to perform front-running and related attacks, such as
{\em sandwich attacks}.

Our second contribution is to formalize this notion through two
correctness conditions: {\em $\epsilon$-Ordering Equality} and {\em
  $\Delta$-Ordering Separation}.  Informally, $\epsilon$-Ordering
Equality requires that, for requests issued at the same time, the
probability of any permutation deviates from the uniform distribution
by at most $\epsilon$.  {\em $\Delta$-Ordering Separation} ensures
that, for requests submitted at least $\Delta$ time apart, the earlier
request precedes the later one in the ledger.

To enforce these conditions, we introduce the Secret Random Oracle
(SRO), an abstraction that provides a fault-tolerant and unbiased
source of randomness.  The SRO is used to inject controlled randomness
into the final order derived from replicas' ordering indicators.  We
present two SRO designs: one based on Trusted Execution Environments
(TEEs) and one based on threshold verifiable random
functions~\cite{christian2000Random, galindo2021fully}.
Because the profitability of front-running and sandwich attacks
depends on obtaining a specific ordering of requests, introducing
randomness reduces the likelihood of adversarially preferred outcomes
and thereby mitigates them.

Our third contribution is to characterize the trade-off between
$\epsilon$-Ordering Equality and $\Delta$-Ordering Separation.
Ideally, both $\epsilon$ and $\Delta$ would be small.  However, while
increasing randomness mitigates the attacks, it can also compromise
the role of a request's  {\em relevant} features in determining
its position in the ledger: thus, decreasing $\epsilon$ comes at the
price of increasing the $\Delta$ required to ensure that requests
separated in time are ordered correctly  despite
Byzantine behavior.

Finally, we explore the practical implications of this trade-off by
designing, implementing, and evaluating \sys, a new ordered consensus
protocol for the {\em partially synchronous}
model~\cite{dwork88consensus} introduced to circumvent the
impossibility of achieving safe, live, and fault-tolerant consensus in
fully asynchronous systems~\cite{fischer83impossibility}.  \sys is
always safe and guarantees liveness during sufficiently long periods
of synchrony (formally, after some unknown Global Stabilization Time).
During such periods, \sys enforces both $\epsilon$-Ordering Equality
and $\Delta$-Ordering Separation.
\sys extends Pomp\={e} by injecting SRO-generated noise into the
fault-tolerant timestamps used to order client requests.

Our evaluation shows that \sys can be effective in mitigating
front-running and sandwich attacks while incurring moderate
performance overhead.  For example, when adding random noise sampled
from $[0, 5\Delta_{net}]$, where $\Delta_{net}$ bounds message delay
during synchronous periods ({\em e.g.,} $\Delta_{net}=400$ms in Solana
Alpenglow~\cite{solanaalpenglow}), \sys maintains $\epsilon$ below
$5\%$---a threshold considered acceptable in other equal-opportunity
contexts~\cite{originalempolymentact,empolymentact}---while matching
Pomp\={e}'s throughput.  The added noise increases the median and
99th-percentile latency by $1.12\times$--$1.42\times$ in a
geo-distributed deployment spanning 12 cities.

\section{Equal opportunity}
\label{s:properties}

Consider a system in which clients invoke commands, and nodes ({\em i.e.,} servers)
aim to produce a total order that reflects the
time at which commands are invoked---earlier commands should precede
later ones.

\subsection{Motivating equal opportunity}
\label{s:violate-ordering-equality}

Informally, if two commands have the same invocation time, equal
opportunity requires the two possible orders to be equally likely in
the system output. Similarly, if three commands are invoked at the
same time, all six possible orders should be equally likely.
To illustrate how equal opportunity is violated in practice, we analyze
publicly available traces from Ethereum.

\paragraph{Bias under two concurrent invocations.}
Violations of equal opportunity with two concurrent invocations not
only indicate bias but can also enable front-running~\cite{daian20flash, torres2021frontrunner}.
Empirical studies show that both phenomena have significantly affected
the allocation of \$89M over 32 months in the Ethereum
blockchain~\cite{qin2022quantifying}.
Although Ethereum has since transitioned from proof-of-work to
proof-of-stake~\cite{ethmerge}, the ordering problem we study depends
only on the semantics of client requests; we defer a discussion of
proof-of-stake to Section~\ref{s:impl}.

In Ethereum, an invocation originating in Europe is more likely to
precede a simultaneous one from Australia because a larger fraction of
Ethereum nodes are located in Europe.
If the system orders European invocations before Australian ones more
than half the time, we say it is \textit{biased} toward Europe.
Geographic location is an {\em irrelevant} feature (a protected
characteristic) under equal-opportunity legislation~\cite{civilrightact},
yet such bias has been observed in Ethereum and other
blockchains~\cite{liu23flexible}.

This bias can lead to significant consequences in applications such as
blockchain liquidations.
In traditional finance, liquidation occurs when an asset is sold below
market value to repay debt, creating profit opportunities for buyers.
Similar mechanisms exist in blockchains, where they
provide a common profit mechanism in the stablecoin~\cite{makerdao}
and lending\cite{aave} applications. The buyer whose command is ordered
first on the blockchain is typically the one to realize the profit.

Consider two clients, one in Europe and one in Australia, invoking a
liquidation simultaneously, trying to secure a \$200K
profit.  If the system orders the European request first with
probability $0.75$, then the expected profit is \$0.2M $\times$ 0.75 =
\$150K for the European client and \$0.2M $\times$ 0.25 = \$50K for
the Australian client.  Thus, irrelevant features can lead to
substantial disparities among clients that should be treated
equally. Similarly, a client intent on becoming the beneficiary of a 
liquidation's profits could leverage faster network devices to further
bias ordering and front-run other clients.

\begin{figure}[]
  \fontsize{11pt}{13pt}\selectfont
\begin{tabularx}{1\columnwidth}{|l|l|l|}
\hline
\begin{tabular}[c]{@{}l@{}}Output order\end{tabular} & \begin{tabular}[c]{@{}l@{}}Victim's profit\end{tabular} & \begin{tabular}[c]{@{}l@{}}Attacker's profit\end{tabular}
\\ \hline
 $i_2, \textcolor{Bittersweet}{i_1}, i_3$
&  -\$500
&  \$800 
\\ \hline
$i_3, \textcolor{Bittersweet}{i_1}, i_2$
&  \$700
&  -\$400
\\ \hline
other order
&  \$300
&  \$0 \\ \hline
\end{tabularx}
\caption{An example of a sandwich attack in which a victim invokes $i_1$ and an attacker invokes $i_2$ and $i_3$.
The example is drawn from a decentralized exchange; details of the associated values are explained in Section~\ref{s:eval:sandwich}.}
\label{fig:sandwich}
\vspace{-5px}
\end{figure}

\paragraph{Three-invocation sandwich attack.}
Violations of equal opportunity under three simultaneous invocations 
enable sandwich attacks~\cite{zhou2021high}.
Empirical studies estimate such attacks extracted more than
\$174M over 32 months on Ethereum~\cite{qin2022quantifying}.

Figure~\ref{fig:sandwich} illustrates a sandwich attack observed in a decentralized
exchange.
After a victim submits command $i_1$, an attacker submits $i_2$ and $i_3$.
The attacker profits only if the system outputs the order
$i_2, i_1, i_3$.
Thus, a successful attack relies on making this particular ordering
significantly more likely than equal opportunity would permit.
A common strategy is to privately relay $i_2$ and $i_3$ to colluding
nodes that then exclusively propose blocks containing the order $i_2, i_1,
i_3$ ~\cite{qin2022quantifying}.
While attackers are free to choose their trading strategy, and
different strategies may lead to different expected profits,
the system should not allow them to influence the probabilities of its
possible outputs. Under equal opportunity, all
six permutations of the three commands should occur with equal (or near-equal)
likelihood.

\subsection{A model for equal opportunity}
\label{s:two-principles}

We model equal opportunity in terms of two classical principles from the
economics literature~\cite{young1995equity}: {\em impartiality} and
{\em consistency}.  In our setting, impartiality captures the idea
that irrelevant features ({\em e.g.}, client geolocation) should not affect
the order of commands, while consistency requires that introducing new
commands should not alter the relative ordering of existing ones.
Together, these principles suggest a natural mechanism for enforcing
equal opportunity: a {\em point system}.
We now formalize these principles in the client-server setting and the
ordered consensus framework~\cite{zhang2020byzantine}.

\paragraph{Client invocation.}
An \textit{invocation} is a pair $\langle c, \vec{f_r} \rangle$, where
$c$ is a command and $\vec{f_r}$ is a vector of relevant features,
{\em i.e.,} the only attributes that should influence the command’s
position in the order.  Irrelevant features include, for example,
client identity, geographic location, wealth, and network
connectivity.  In blockchain settings, relevant features typically
include invocation time and transaction fee.  An \textit{invocation
  profile}, denoted $\Io$, is a set of invocations.

\paragraph{Node preference.}
Given an invocation profile $\Io$, each node observes the invocations
and expresses a \textit{preference}.  A preference is a set of pairs
$\langle i, o \rangle$, where $i$ is an invocation and $o$ is an
ordering indicator ({\em e.g.}, a timestamp or score).  Intuitively,
ordering indicators encode how a node would rank invocations.  A
\textit{preference profile}, denoted $\Po$, is a vector of preferences
from all correct nodes.  We say that $\Po$ is \textit{well-formed}
under $\Io$ if it references only invocations in $\Io$.

\paragraph{Chance relation.}
In practice, due to nondeterminism ({\em e.g.}, network delays or
scheduling effects), a given set of invocations may result in
different preference profiles.  We assume a ground-truth invocation
profile $\Io_{gt}$ and model the uncertainty in the preferences of
correct nodes through \textit{chance relations}.  For two preference
profiles $\Po_1$ and $\Po_2$ well-formed under $\Io_{gt}$, $\Po_1
\succ_c \Po_2$ denotes that $\Po_1$ is more likely to occur than
$\Po_2$, and $\Po_1 \sim_c \Po_2$ denotes that they occur with equal
likelihood.
Note that $\Po_1 \sim_c \Po_2$ is equivalent to $\neg (\Po_1 \succ_c \Po_2) \wedge \neg (\Po_2 \succ_c \Po_1)$.

\paragraph{Impartiality.}
Impartiality requires that invocations with identical relevant
features be treated symmetrically.  Formally, for any $i_1, i_2 \in
\Io_{gt}$ with identical $\vec{f_r}$, a system is \textit{impartial}
if and only if swapping $i_1$ and $i_2$ in any preference profile does
not change its likelihood.  That is, if $\Po_2$ is obtained from
$\Po_1$ by swapping $i_1$ and $i_2$, then $\Po_1 \sim_c \Po_2$.

\begin{figure}[]
 \centering
 \fontsize{11pt}{13pt}\selectfont
\begin{align*}
\Io_{gt}  &= \{ \langle{}c_1, \langle{}\text{5pm}\rangle{}\rangle{}_{i_1}, \langle{}c_2, \langle{}\text{5pm}\rangle{}\rangle{}_{i_2}, \langle{}c_3, \langle{}\text{5:01pm}\rangle{}\rangle{}_{i_3} \} \\
\Po_1 &= \langle{}\{ \textcolor{Bittersweet}{\langle{}i_1, 1\rangle{}}, \textcolor{Bittersweet}{\langle{}i_2, 2\rangle{}}, \langle{}i_3, 3\rangle{}\}\rangle{} \\
\Po_2 &= \langle{}\{ \textcolor{Bittersweet}{\langle{}i_2, 1\rangle{}}, \textcolor{Bittersweet}{\langle{}i_1, 2\rangle{}}, \langle{}i_3, 3\rangle{}\}\rangle{}
\end{align*}
\vspace{-24px}
\caption{In a system that uses invocation time as the only relevant feature and sequence numbers as ordering indicators, impartiality requires preference profiles $\Po_1$ and $\Po_2$ to be equally likely.}
\vspace{-2px}
\label{fig:property-example1}
\end{figure}

Impartiality is the first pillar of equal opportunity and it is illustrated in 
Figure~\ref{fig:property-example1}.  If
$i_1$ and $i_2$ share the same relevant feature ({\em e.g}.,
invocation time), then their relative order must be independent of
irrelevant features and $\Po_1$ and  $\Po_2$  must be equally likely.

\paragraph{Consistency.}
Consistency, the second pillar of equal opportunity, requires that the
relative ordering of invocations depend only on their own features and
not on the presence or absence of other invocations.
Figure~\ref{fig:property-example2} illustrates this principle: the key
idea is that the relative ordering of $i_1$ and $i_2$ should be
unaffected by the presence or absence of a third invocation
$i_3$. Consistency requires $\Po_1 \succ_c \Po_2 \iff \Po_3
\succ_c \Po_4$: the order of $i_1$ and $i_2$ should depend solely on
their \textit{own} relevant features.

Let $\Po_1 =_{\Io} \Po_2$ denote that $\Po_1$ and $\Po_2$ agree on
invocations in $\Io$, and let $\Io_{gt} = \{\Io_1, \Io_2\}$ denote a
partition.  A system is \textit{consistent} if and only if, for all preference
profiles $\Po_1, \Po_2, \Po_3, \Po_4$ well-formed under $\Io_{gt}$,
whenever $\Po_1 =_{\Io_1} \Po_3$, $\Po_2 =_{\Io_1} \Po_4$, $\Po_1
=_{\Io_2} \Po_2$, and $\Po_3 =_{\Io_2} \Po_4$, it holds that $\Po_1
\succ_c \Po_2 \iff \Po_3 \succ_c \Po_4$.

\paragraph{Point systems.}
A natural way to satisfy both principles is through a \emph{point
system}.  In such a system, each invocation is assigned a score that
depends only on its relevant features.  Invocations are ordered by
their scores, breaking ties uniformly at random.

A point system satisfies impartiality because identical relevant
features yield identical scores, leading to symmetric treatment.  It
satisfies consistency because each score depends only on the
invocation itself, and is therefore independent of other invocations.

Indeed, in a different model, point systems have been shown to be the
only mechanisms that satisfy both impartiality and
consistency~\cite{young1995equity}.  Establishing an analogous result
in the ordered consensus model remains an interesting direction for
future work.

\begin{figure}[]
  \centering
  \fontsize{11pt}{13pt}\selectfont
\begin{align*}
\Io_{gt} &= \{ \langle{}c_1, \text{1pm}\rangle_{i_1}, \langle{}c_2, \text{1:01pm}\rangle_{i_2}, \langle{}c_3, \text{1:05pm}\rangle_{i_3} \} \\
\Po_1 &= \langle{}\{ \textcolor{Bittersweet}{\langle{}i_1, 1\rangle, \langle{}i_2, 2\rangle} \}\rangle
\;\; \Po_3 = \langle{}\{ \textcolor{Bittersweet}{\langle{}i_1, 1\rangle, \langle{}i_2, 2\rangle}, \textcolor{OliveGreen}{\langle{}i_3, 3\rangle} \}\rangle\\
\Po_2 &= \langle{}\{ \textcolor{RoyalPurple}{\langle{}i_1, 2\rangle, \langle{}i_2, 1\rangle} \}\rangle
\;\; \Po_4 = \langle{}\{ \textcolor{RoyalPurple}{\langle{}i_1, 2\rangle, \langle{}i_2, 1\rangle}, \textcolor{OliveGreen}{\langle{}i_3, 3\rangle} \}\rangle
\end{align*}
\vspace{-15px}
\caption{An example of consistency where $\Io_1 = \{ i_1, i_2 \}$ and $\Io_2 = \{ i_3 \}$.
This example originates from real-world legislation on equal opportunity in the context of
kidney exchange, as explained in Section~\ref{s:related}.
}
\vspace{-2px}
\label{fig:property-example2}
\end{figure}

\subsection{Correctness conditions for systems}
\label{s:two-properties}

In computer systems, invocation time is the primary relevant feature:
commands are intended to be ordered according to {\em when} they are invoked.
In principle, equal opportunity can be enforced by a point system that
uses the invocation time as the score.

In practice, however, invocation time cannot be measured precisely.
Nodes can only observe when they receive a command, and these observations
are affected not only by the true invocation time but also by irrelevant
features such as network latency and geographic location.
As a result, a perfect point system based on exact invocation times is unattainable.

To accommodate this uncertainty, we relax the point-system ideal using
two parameters, $\epsilon$ and $\Delta$, and define the corresponding
correctness conditions: \textit{$\epsilon$-Ordering Equality} and
\textit{$\Delta$-Ordering Separation}. These can be viewed
as approximate counterparts of  earlier principles:
$\epsilon$-Ordering Equality relaxes impartiality by allowing small
deviations from uniform ordering, while $\Delta$-Ordering Separation
relaxes consistency by tolerating bounded uncertainty in invocation
times.

\paragraph{$\epsilon$-Ordering Equality.}
This condition captures approximate fairness among invocations with identical
relevant features.
For any subset $\Io \subseteq \Io_{gt}$ in which all invocations share the same
invocation time, let $n$ denote the number of invocations in $\Io$.
For any total order $\succ$ over $\Io$,  $|Pr[\succ]- \frac{1}{n!}| \leq \epsilon(n)$.

Here, $\epsilon(n)$ bounds how far the distribution over permutations
can deviate from uniform.  When $\epsilon = 0$, all permutations are
equally likely, recovering impartiality.  More generally, smaller
values of $\epsilon$ ensure that invocations with identical relevant
features have similar chances of appearing in any position,
approximating how  a point system breaks ties.

\paragraph{$\Delta$-Ordering Separation.}
This condition captures the ability to respect differences in
invocation time. It requires that for any invocations $i_1, i_2 \in
\Io_{gt}$, if $i_1$ is invoked at least $\Delta$ time units before
$i_2$, then $i_1$ precedes $i_2$ in the system output.

Intuitively, $\Delta$ represents the minimum separation required for the system
to reliably distinguish between two invocations despite noise in timing observations.
Smaller values of $\Delta$ correspond to stronger guarantees that earlier
invocations appear earlier in the output.

Ideally, one would like $\epsilon = 0$ and $\Delta = 0$, in which case both
perfect fairness and perfect temporal ordering are achieved.
This ideal is attainable (alas, only) in an idealized setting where all nodes are correct
and observe invocation times without error, as captured by the following theorem.

\begin{theorem}
  \label{ordering-equality-simple}
  If all nodes are correct and accurately measure the invocation time of all invocations, ordering equality and ordering separation with $\epsilon=\Delta=0$ can be achieved using a point system.
\end{theorem}

\noindent{\em Proof.} 
Assuming all observations are accurate to a certain time unit ({\em e.g.}, microsecond), a straightforward approach is to assign each invocation a score equal to its observed invocation time.
In such a point system, we have: (1) all invocations are sorted by this score in ascending order; and (2) the order of invocations with the same score is chosen uniformly at random from all permutations.
(1) guarantees $\Delta = 0$ because invocations with different observed invocation times are guaranteed to be ordered in the order they are invoked.
(2) guarantees $\epsilon = 0$ because any two total orders have the same probability when chosen uniformly at random.
\pushQED{\qed}
\qedhere \popQED

\subsection{Mitigating ordering-based attacks}
A system that enforces equal opportunity can significantly mitigate
ordering-based attacks, though it does not necessarily make them
completely unprofitable.  In existing systems, victims of
front-running attacks are likely to incur losses because the system is
biased toward clients with lower network latency.  Under equal
opportunity, this advantage is reduced: the gap between an attacker
(who may have superior network access) and a victim narrows, as the
system treats competing requests more uniformly. The ideal case of
equal opportunity, with $\epsilon=\Delta=0$, provides the strongest
possible mitigation that is agnostic to the semantics of client
commands.  Further reducing an attacker's advantage would require
identifying and treating certain commands as adversarial, which risks
introducing new sources of bias.

\section{Secret Random Oracle}
\label{s:design}

\begin{figure}[]
  \centering
  \includegraphics[width=0.45\textwidth]{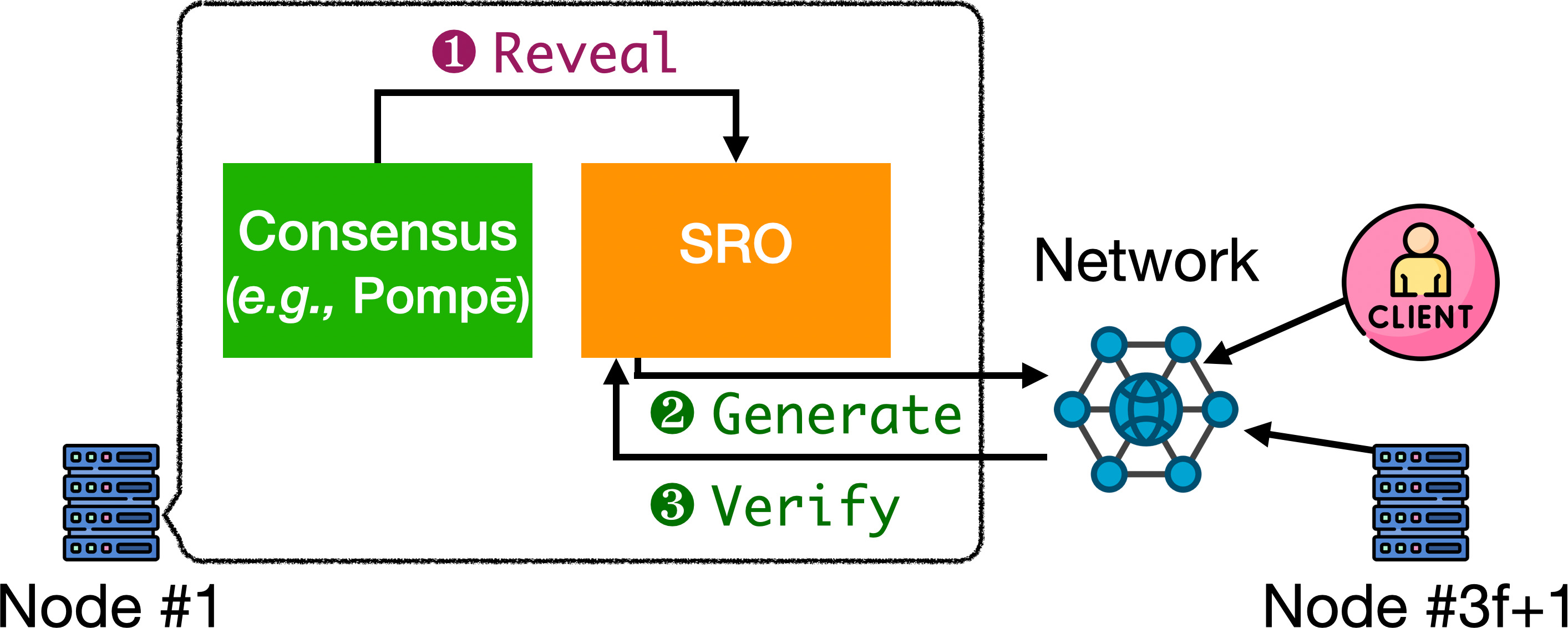}
  \caption{The protocol design overview of \sys. The consensus component invokes the \texttt{Reveal} interface, which returns a random number. The TEE-based SRO returns the random number immediately. The TVRF-based SRO sends a message generated by the \texttt{Generate} interface and verifies messages received from other nodes. After receiving enough valid messages, the SRO will construct and return a random number to the consensus component. }
  \label{fig:overview}
  \vspace{-10px}
\end{figure}

The point-system abstraction suggests that randomness is essential for
achieving ordering equality.  Crucially, this randomness must remain
hidden until after other system components commit to their outputs;
otherwise, adversaries could exploit it to bias ordering. We address
this requirement with the \textit{Secret Random Oracle} (SRO), a
system component that generates random values and reveals them only
after a coordinated commit, preventing ordering decisions from
depending on them.

\paragraph{Design overview.}
Figure~\ref{fig:overview} shows how the SRO integrates with a consensus system.
We consider a BFT system with $n = 3f + 1$ nodes running a protocol
such as Pomp\={e}~\cite{zhang2020byzantine}.  When using a
point-system approach, the consensus component invokes the SRO to
obtain randomness via the \texttt{Reveal} interface.

We present two realizations of the SRO abstraction.  A TEE-based
design generates and reveals randomness locally without network
communication.  A cryptographic design based on Threshold Verifiable
Random Functions (TVRF) requires communication: nodes invoke
\texttt{Generate} to produce messages, use \texttt{Verify} to validate
them, and combine sufficiently many valid messages to recover the
random value.

The remainder of this section details the SRO interface and its
guarantees (\S\ref{s:design-sro-interface}), presents two concrete
realizations based on trusted hardware (\S\ref{s:design-sro-tee}) and
cryptography (\S\ref{s:design-sro-crypto}), shows how SRO integrates
into \sys to enforce the desired ordering
properties(\S\ref{s:design-integration}), and quantifies the fundamental trade-off
that exists in trying to achieve both $\epsilon$-Ordering Equality and
$\Delta$-Ordering Separation(\S~\ref{s:design-quantify}).

\subsection{SRO interface and guarantees}
\label{s:design-sro-interface}

The SRO provides a simple abstraction for generating and revealing
random values in a coordinated and adversary-resistant manner.

We consider a system with $n$ nodes, at most $f$ of which may be
Byzantine.  Each node holds a private key and knows the public keys of
all others. A function (not shown) maps an integer $k$ to a
pseudorandom number. The SRO interface, shown in
Figure~\ref{fig:arch}, consists of three operations:

\begin{itemize}

\item \texttt{Reveal(k, s)} returns the output of the pseudorandom
  function evaluated at  $k$, once $s$ distinct signatures
  indicate a quorum of nodes wishes to disclose it.
\item \texttt{Generate(k)} takes in $k$ and produces a cryptographic proof.
\item \texttt{Verify(k, p, r)} given a proof $p$, checks whether the value $r$ returned when invoking \texttt{Reveal} on $k$ is correct.
\end{itemize}

\texttt{Reveal(k)} returns its output only after enough nodes signal their readiness to do so.
This ensures that randomness remains hidden until it can no longer influence earlier decisions.

The SRO provides the following guarantees.

\paragraph{Uniqueness.} All valid invocations of \texttt{Reveal(k)} return the same value.
A set of signatures is valid if it contains valid signatures on $k$
from at least $n-f$ distinct nodes.  Invalid inputs cause
\texttt{Reveal} to return an error.

\paragraph{Secrecy.}  For any $k$, let $r$ be the value returned by \texttt{Reveal(k)}.
If an adversary does not obtain valid signatures from $n-f$ nodes,
it is computationally infeasible to distinguish  with non-negligible probability $r$ from a uniform random value.

\paragraph{Randomness.} For all $k$, the value $r$ returned by \texttt{Reveal(k)} is a non-error uniform random sample from its codomain.

\paragraph{Validity.} If \texttt{Generate(k)} outputs proof $p$ and \texttt{Reveal(k)} returns $r$,
then \texttt{Verify(k, p, r)} returns \texttt{True},
and it is computationally infeasible to find some integer  $r' \neq r$ such that
\texttt{Verify(k, p, r')} also returns \texttt{True}.

\begin{figure}[]
\fontsize{10pt}{12pt}\selectfont
\begin{tabular*}{1\columnwidth}{@{\extracolsep{\fill}}l r}
\toprule

\texttt{\textcolor{RoyalPurple}{Reveal}(Int \textcolor{Bittersweet}{k}, Set<Signature> \textcolor{Bittersweet}{s})} $\rightarrow$ \texttt{Int | Error}\\

\texttt{\textcolor{RoyalPurple}{Generate}(Int \textcolor{Bittersweet}{k})} $\rightarrow$ \texttt{Proof}\\

\texttt{\textcolor{RoyalPurple}{Verify}(Int \textcolor{Bittersweet}{k}, Proof \textcolor{Bittersweet}{p}, Int \textcolor{Bittersweet}{r})} $\rightarrow$ \texttt{Bool}\\

\bottomrule
\end{tabular*}

\caption{The interface of a Secret Random Oracle (SRO).}
\label{fig:arch}
\vspace{-5px}
\end{figure}

\subsection{An SRO design using trusted hardware}
\label{s:design-sro-tee}

Trusted Execution Environments (TEEs) provide secure enclaves that
protect the integrity and confidentiality of code and data.
They have already been used in blockchain systems such as Ethereum's
Sepolia testnet~\cite{thirdteeblog} and Solana's Block Assembly Marketplace~\cite{solanateeordering}.
Here, we show how they can be used to implement an SRO.

\paragraph{Initialization.}
Each node hosts a TEE running the SRO logic.  During initialization,
TEEs generate candidate random values ({\em e.g.}, via \texttt{RDRAND}
in x86) and run a consensus protocol to agree on a shared secret one.
This number remains confined within the TEEs and will be used as the \texttt{seed} of a pseudorandom 
function denoted as \texttt{RAND}.

\paragraph{Reveal, Generate and Verify.}
To invoke \texttt{Reveal},  a node forwards an  integer and a set of signatures to its
local TEE, which returns \texttt{RAND(seed, k)} if enough valid signatures are present; otherwise, it returns an error.
Similarly, \texttt{Generate} forwards \texttt{k} to the TEE, which returns \texttt{HASH(RAND(seed, k))} where \texttt{HASH} is a one-way function.
Lastly, \texttt{Verify} checks that its parameter \texttt{p} equals
\texttt{HASH(r)}.

\paragraph{Correctness.}
Determinism of \texttt{RAND} ensures uniqueness,
while its pseudorandomness guarantees randomness.
Secrecy follows from TEE confidentiality,
and validity from the collision resistance of the hash function.

We assume that initialization ({\em i.e.,} consensus on the random seed) eventually completes, thereby 
ensuring liveness: all invocations of SRO operations on correct nodes eventually terminate.

\subsection{An SRO design using threshold VRF}
\label{s:design-sro-crypto}

A Threshold Verifiable Random Function (TVRF) is a cryptographic primitive
that enables a set of nodes to jointly evaluate a pseudorandom function
in a distributed and fault-tolerant manner~\cite{christian2000Random, galindo2021fully}.
Such constructions have been widely used in Byzantine agreement protocols
for tasks such as committee selection and randomized consensus.

In our setting, a TVRF provides a natural way to implement the SRO abstraction:
random values are generated collaboratively and remain unpredictable
until sufficiently many nodes agree to disclose them.

Let \texttt{TVRF} denote a pseudorandom function mapping an integer $k$
to a pseudorandom value.
Figure~\ref{fig:tvrf} shows the interface for evaluating \texttt{TVRF(k)}.
Each node invokes \texttt{Produce(k)} to generate a \textit{share} using its private key.
After collecting enough valid shares, a node invokes \texttt{Combine}
to recover \texttt{TVRF(k)}.
To handle Byzantine behavior, \texttt{Valid} checks the correctness of each share.
TVRFs provide the following guarantees~\cite{christian2000Random}.

\paragraph{Robustness.}
For all integers \texttt{k},
it is computationally infeasible for an adversary to produce enough
valid shares such that the integer output of \texttt{Combine} is not \texttt{TVRF(k)}.

\paragraph{Unpredictability.}
Without enough valid shares for \texttt{TVRF(k)}, an adversary 
cannot distinguish \texttt{TVRF(k)} from a uniform random sample with non-negligible probability.

We make a slight modification to threshold VRF.
For the \texttt{Produce} interface, we add a parameter: a set of signatures of \texttt{k} to be verified.
If verification fails, correct nodes must return an error instead of a share.
We can now design an SRO as follows:
\texttt{Reveal} forwards the two parameters to all nodes,
collects enough valid shares, and invokes \texttt{Combine}.
\texttt{Generate} returns the set of public keys.
\texttt{Verify}
takes all the public keys and a set of shares as input and returns whether all the shares are valid (using the \texttt{Valid} interface).

\paragraph{Correctness.}
Under the random oracle model, a threshold VRF outputs cryptographically secure pseudorandom numbers.
Robustness implies uniqueness and validity: for all integers \texttt{k}, combining enough valid shares can only produce \texttt{TVRF(k)}, since an adversary cannot create valid shares leading to a combined value different from \texttt{TVRF(k)}.
Unpredictability implies secrecy because an adversary cannot obtain enough valid shares without valid signatures of \texttt{k} out of $n-f$ nodes, and without enough valid shares, it has no
 information about \texttt{TVRF(k)}.
Liveness is ensured, assuming all network messages are eventually delivered.

\begin{figure}[]
\fontsize{9.9pt}{12pt}\selectfont
\begin{tabular*}{1\columnwidth}{@{\extracolsep{\fill}}l r}
\toprule

Threshold VRF node-side function\\

\texttt{\textcolor{RoyalPurple}{Produce}(Int \textcolor{Bittersweet}{k})} $\rightarrow$ \texttt{Share}\\

Threshold VRF client-side functions\\

\texttt{\textcolor{RoyalPurple}{Combine}(Set<Share> \textcolor{Bittersweet}{s})} $\rightarrow$ \texttt{Int | Error}\\

\texttt{\textcolor{RoyalPurple}{Valid}(Int \textcolor{Bittersweet}{k}, Int \textcolor{Bittersweet}{node\_id}, Share \textcolor{Bittersweet}{s})} $\rightarrow$ \texttt{Bool}\\

Modified node-side function for SRO\\

\texttt{\textcolor{RoyalPurple}{Produce}(Int \textcolor{Bittersweet}{k}, Set<Signature> \textcolor{Bittersweet}{s})} $\rightarrow$ \texttt{Share | Error}\\

\bottomrule
\end{tabular*}

\caption{The interface of (modified) threshold VRF.}
\label{fig:tvrf}
\vspace{-5px}
\end{figure}

\subsection{Integrating an SRO with Pomp\={e}}
\label{s:design-integration}

We now introduce \sys, which integrates an SRO with Pomp\={e}~\cite{zhang2020byzantine}, a state-of-the-art ordered consensus protocol.
The goal is to enforce the conditions in Section~\ref{s:two-properties}.

\paragraph{The Pomp\={e} protocol.}
Pomp\={e} employs any standard leader-based BFT SMR protocol ({\em e.g.,} ~\cite{castro02practical}) that offers a primitive to agree on a value for each slot in a sequence of consensus decisions.
Pomp\={e} transforms such a protocol into a new one that enforces correctness conditions on ordering.

Specifically, Pomp\={e} associates the slots with consecutive time intervals.
For example, one slot may be associated with time interval $[t,~t + 500ms)$, and the next slot could be associated with interval $[t + 500ms,~t + 1000ms)$.
For simplicity, Pomp\={e} assumes that such a mapping from slots to time intervals is common knowledge.
Within each consensus slot, the value to agree on is a set of $\langle c,~ats \rangle$ pairs where $c$ is a command and $ats$ is called an \textit{assigned timestamp}.
Pomp\={e} provides two guarantees for this assigned timestamp:
(1) $ats$ falls in the time interval associated with the slot;
(2) $ats$ is bounded by the lowest and highest timestamps
provided by correct nodes. The commands are then ordered by
their assigned timestamps.

Pomp\={e} requires $3f+1$ nodes and ensures guarantee (2) by
collecting $2f+1$ timestamps from different nodes for each command
$c$.  The \textit{median} of the $2f+1$ timestamps is chosen as the
assigned timestamp $ats$ for command $c$.  Since there are at most $f$
faulty nodes, the median of any $2f+1$ timestamps is upper-bounded and
lower-bounded by timestamps provided by correct nodes.

\paragraph{Integrating Pomp\={e} with an SRO.}
\sys augments Pomp\={e} by adding controlled randomness to the
assigned timestamps of each command.  After consensus is reached for
slot \texttt{k}, a node obtains a quorum of signatures certifying the
decision.  Using \texttt{k} and these signatures, a correct node
invokes \texttt{Reveal} and obtains a random seed to be used by the pseudorandom
number generator. Crucially, this seed remains unknown until consensus
is finalized, ensuring that it is independent of the consensus decision.

The pseudorandom number generator assigns an independent value $r$ to
each command, which is used to sample from a distribution $D$ over
$[0,~\Delta_{noise}]$. Commands are then ordered by $ats +
\mathtt{sample}(D, r)$ instead of $ats$ alone.  Section
\ref{s:design-quantify} discusses how to select the distribution and
$\Delta_{noise}$.

\paragraph{Stability.}
A command is \textit{stable} (or finalized) when it appears in the
output ledger.  In Pomp\={e}, stability coincides with reaching
consensus.  In \sys, the added random noise may delay stability.  Let
the latest finalized slot correspond to interval $[ts, ts')$.  A
  command $c$ becomes stable when $ats + \texttt{sample}(D, r) < ts'$,
  ensuring that $c$ and all preceding commands can be safely added to
  the output ledger.

\paragraph{Safety and liveness.} Pomp\={e} guarantees the same safety and liveness properties as
classic SMR protocols~\cite{castro02practical,lamport98parttime}.
\sys inherits these guarantees and differs only in how commands are
ordered.  We now turn to proving the corresponding correctness
conditions on ordering.

\subsection{A trade-off between $\epsilon$-Ordering Equality and $\Delta$-Ordering Separation}
\label{s:design-quantify}

The key question is how to choose the noise range $\Delta_{noise}$ and
the distribution $D$ used for random perturbation.
We first derive a lower bound on $\epsilon(2)$ for any fixed $\Delta$,
as a function of $\Delta_{noise}$ and the system parameter $\Delta_{net}$,
which bounds message delay, processing time, and clock drift under
partial synchrony~\cite{dwork88consensus}.

\paragraph{Partial synchrony model.} 
One variant of the partial synchrony model introduces the \emph{Global
Stabilization Time} (GST) \cite{dwork88consensus}.  Specifically,
there is an unknown time GST such that, after this time, there is a
known bound $\Delta_{net}$ on network latency and processing time.
The safety and liveness of Pomp\={e} are proven under this model.  We
now analyze the ordering properties of \sys in the same model.  More
precisely:

\begin{assumption}
  \label{assumption-sync-net}
After GST,
if a command is invoked at time $T$, correct nodes will provide timestamps in the range $[T, T + \Delta_{net})$ for this command.
\end{assumption}

Note that a simple clock synchronization protocol has
been given as part of the Pomp\={e} protocol.

We now prove that any probability distribution $D$ on $[0, \Delta_{noise}]$ achieves $\Delta$-Ordering Separation with $\Delta = \Delta_{net} + \Delta_{noise}$ and,
in the best case, provides $\epsilon$-Ordering Equality with $\epsilon(2) \ge \frac{1}{2(k + 1)}$ 
($k$ is the natural number such that $k\Delta_{net} \le \Delta_{noise} < (k + 1)\Delta_{net}$).
We also construct a discrete probability distribution that achieves this lower bound.
These results establish a fundamental trade-off:
improving ordering equality (reducing $\epsilon$) necessarily worsens
ordering separation (increasing $\Delta$).

\begin{lemma}
  \label{lemma-assigned-timestamp}
  The assigned timestamp of a command is bounded by timestamps provided by correct nodes.
\end{lemma}
\noindent{\em Proof.}  See~\cite{zhang2020byzantine}.

\begin{theorem}($\Delta$-Ordering Separation)
  \label{theorem-ordering-separation}
After GST, for all invocations $i_1$ and $i_2$, if $i_1$ is invoked more than $\Delta_{net}+\Delta_{noise}$ time units earlier than $i_2$, then $i_1$ is guaranteed to appear before $i_2$ in the output.
\end{theorem}

\noindent{\em Proof.}
Suppose $i_1$ and $i_2$ are invoked at time $T_1$ and $T_2$.
By Assumption~\ref{assumption-sync-net} and Lemma~\ref{lemma-assigned-timestamp},
the assigned timestamp of $i_1$ is in the range $[T_1, T_1 + \Delta_{net})$.
After adding the random noise, the resulting timestamp is in the range $[T_1, T_1 + \Delta_{net} + \Delta_{noise})$.
Similarly, the resulting timestamp for $i_2$ is in the range $[T_2, T_2 + \Delta_{net} + \Delta_{noise})$.
Therefore, if $T_2 > T_1 + \Delta_{net} + \Delta_{noise}$, $i_2$ will appear after $i_1$ in the output.
\pushQED{\qed}
\qedhere \popQED

\begin{theorem}($\epsilon(2)$-Ordering Equality lower bound)
  \label{theorem-ordering-equality-2}
After GST,
for any probability distribution $D$ on $[0, \Delta_{noise}]$,
for all invocations $i_1$ and $i_2$ invoked at the same time,
$|Pr[i_1 \prec i_2 ] - \frac{1}{2}| \ge \frac{1}{2(k + 1)}$
where $k$ is the natural number such that $k\Delta_{net} \le \Delta_{noise} < (k + 1)\Delta_{net}$.
\end{theorem}

\noindent{\em Proof.}
The key idea is that if two random samples fall within the same
$\Delta_{net}$-sized interval, their relative order can be determined
by adversarially chosen timestamp differences.

Suppose $i_1$ and $i_2$ are both invoked at time $T$.
By Assumption~\ref{assumption-sync-net} and Lemma~\ref{lemma-assigned-timestamp},
the assigned timestamps are in the range $[T, T + \Delta_{net})$,
which the adversary may pick.
Let the two samples from $D$ be $s_1, s_2 \in [0, \Delta_{noise}]$.
Let $k$ be the natural number defined above.
We split the interval $[0, \Delta_{noise}]$ into $k + 1$ buckets: $B_0 = [0, \Delta_{net}), B_1 = [\Delta_{net}, 2\Delta_{net}), \dots, B_k = [k\Delta_{net}, \Delta_{noise}]$.
If $s_1$ and $s_2$ fall into the same bucket, then $|s_1 - s_2| < \Delta_{net}$, and the adversary can decide the ordering of the two commands.
Let the probability mass of $D$ in those $k + 1$ buckets be $pb_0, pb_1, \dots, pb_k$.
We know $ \sum_i pb_i  = 1$ and the probability that $s_1$ and $s_2$ fall into the same bucket is $\sum_i pb_i^2$.
By the relation between the arithmetic mean and the quadratic mean, $\sum_i pb_i^2 \ge \frac{1}{k + 1}$.
Thus, the adversary can decide the ordering with at least a probability of $\frac{1}{k + 1}$.
This gives the lower bound $\epsilon(2) \ge \frac{1}{2(k + 1)}$.
\pushQED{\qed}
\qedhere \popQED

\begin{theorem}(Discrete distribution optimality)
  \label{theorem-ordering-equality-3}
	The above lower bound on $\epsilon(2)$ can be achieved by setting $D$ to be a discrete distribution
	that is uniformly random across the discrete points of $0, \Delta_{net}, 2\Delta_{net}, \dots, k\Delta_{net}$.
\end{theorem}

\noindent{\em Proof.}  The proof of the lower bound above also shows
us how to achieve Pareto-optimality, as the arithmetic mean only
equals the quadratic mean if all elements are equal.  This discrete
distribution guarantees this lower bound of $\frac{1}{2(k+1)}$ for
$\epsilon(2)$, as two commands can only be ordered arbitrarily if they
happen to get the same noise, which has a probability of exactly
$\frac{1}{k + 1}$ happening as the distribution has $k + 1$ equally
weighted outcomes.  \pushQED{\qed} \qedhere \popQED

We defer the full proof of the general case $n > 2$ to the appendix.  As suggested in
Section~\ref{s:violate-ordering-equality}, real-world concerns about
ordering could also arise due to violating ordering equality with
three invocations ({\em i.e.,} enforcing $\epsilon(3)$-Ordering
Equality is necessary to mitigate sandwich attacks).

\begin{theorem}($\epsilon(n)$-Ordering Equality)
  \label{theorem-ordering-equality-n}
After GST,
for all invocations $i_1..i_n$ invoked at the same time,
for the discrete distribution constructed as above,
for any total order $\succ$ of $i_1..i_n$,
$|Pr[\succ]-\frac{1}{n!}]| \le \frac{\binom{k + n}{n}}{(k + 1)^n} - \frac{1}{n!}$.
\end{theorem}

\noindent{\em Proof.}  See Appendix~\ref{s:appendix}.

\paragraph{Choosing the $\Delta_{noise}$ parameter.}
In practice, system designers already tune $\Delta_{net}$ in
Pomp\={e} and similar protocols.
To mitigate ordering bias ({\em e.g.}, front-running and sandwich attacks),
they can instead tune $\Delta_{noise}$ to achieve a target $\epsilon$.

For example, when considering attacks involving two or three
invocations, designers may focus on $\epsilon(2)$ and $\epsilon(3)$.
In legal contexts, $\epsilon(2) \approx 0.05$ is often considered acceptable.
The four-fifths rule~\cite{originalempolymentact,empolymentact} requires that
selection probabilities differ by no more than $45\%$ versus $55\%$,
corresponding to $\epsilon(2)=0.05$.
Similar thresholds have influenced fairness definitions in machine
learning ({\em e.g.}, demographic parity~\cite{barocas-hardt-narayanan}).

\begin{figure*}[]

\fontsize{10pt}{12pt}\selectfont
\begin{tabular*}{2.1\columnwidth}{@{\extracolsep{\fill}}l l}
\toprule

State-of-the-art systems trying to provide fairness are all %
vulnerable to bias and front-running. & \S \ref{s:eval:bias}\\

Among all possible financial rewards, an attacker gets the \textit{mean} reward in \sys and the \textit{max} in baselines. & \S \ref{s:eval:sandwich}\\

\sys maintains the same end-to-end throughput as Pomp\={e}, and the latency increases by 1.12x-1.42x. & \S \ref{s:eval:sro}, \ref{s:eval:end-to-end}\\

\bottomrule
\end{tabular*}

\vspace{-5px}
\caption{Table of evaluation results.}
\label{fig:eval-summary}
\vspace{-5px}
\end{figure*}

\section{Implementation}
\label{s:impl}

We implement two SRO variants corresponding to the designs in
Section~\ref{s:design-sro-tee} and Section~\ref{s:design-sro-crypto}.

For the TEE-based variant, we use Intel SGX~\cite{sgx}.
Random values are computed as \texttt{SHA256(seed + k)},
where \texttt{seed} is established during initialization and \texttt{k}
is the input to \texttt{Reveal}.
The SGX SDK provides an optimized \texttt{SHA256} implementation for Linux~\cite{sgxlinux}.
We note that larger hash outputs ({\em e.g.}, \texttt{SHA512}) could provide
stronger security, but are not supported in this environment.
In practice, the integer type used by the SRO need not be limited to the 
standard 4-byte size: blockchain systems commonly use 32-byte or 64-byte values,
which we also support.

For the TVRF-based variant, we build on an existing C++ implementation
of threshold VRF~\cite{galindo2021fully,dvrf} and extend
the \texttt{Produce} interface with signature verification.  Random
values are computed as \texttt{SHA512} hashes of
a combined threshold of shares.
The default configuration of this implementation  uses the \texttt{mcl}
cryptographic library~\cite{mcl} and the BN256 curve.
Unlike in many BFT systems, cryptographic operations do not dominate
the performance overhead in \sys.
As we show in Section~\ref{s:eval}, the overhead is primarily driven by the
added randomness required to balance $\epsilon$ and $\Delta$.
Since cryptographic operations are not the bottleneck, we prioritize
implementation simplicity when selecting cryptographic libraries.

In both variants, signature verification in \texttt{Reveal} relies on
the \texttt{secp256k1} library~\cite{secp256k1},
consistent with the Pomp\={e} implementation.

\paragraph{Discussion}
Permissionless blockchains, pioneered by Bitcoin~\cite{nakamoto08bitcoin},
require nodes to know neither the system size nor the identities of
the participants.  In contrast, \sys follows the standard BFT SMR
model, where nodes know $n$, $f$, and participant public keys.  In
this sense, \sys is not permissionless.
Similarly, systems such as Ethereum and Solana are no longer
permissionless after switching from proof-of-work to proof-of-stake.
\sys is well-suited for such proof-of-stake blockchains.

There is, however, a minor difference between the models used in
proof-of-stake and traditional BFT SMR systems.  With proof-of-stake,
different nodes can \textit{stake} different amounts of native
currency ({\em e.g.,} ETH tokens).  In this setting, the BFT condition
$n = 3f+1$ maps to a bound of at most $33\%$ adversarial stake, while
quorums correspond to $67\%$ of total stake.  We incorporate this
interpretation in our evaluation by adapting baseline protocols
accordingly.

Finally, our current \sys prototype does not support the possibility
that nodes may join or leave the blockchain. However, \sys could easily
be extended to support such dynamic participation by integrating 
standard techniques used for membership changes in traditional SMR
protocols.

\section{Experimental evaluation}
\label{s:eval}

We ask three main questions in our evaluation:
(1) Do the state-of-the-art protocols and \sys address real-world problems like front-running or sandwich attacks?
(2) What do the correctness conditions mean to attackers who try to launch ordering attacks?
(3) What is the end-to-end performance of \sys?
Figure~\ref{fig:eval-summary} shows the answers.

We choose three baselines: HotStuff~\cite{yin19hotstuff}, Pomp\={e}~\cite{zhang2020byzantine}, and Themis~\cite{kelkar2021themis} representing the state-of-the-art fairness notions.
HotStuff adopts the \textit{rotating leadership} notion.
Themis guarantees a fairness property called $\gamma$\textit{-batch-order-fairness},
and we choose $\gamma=1$, which means that if all correct nodes receive $i_1$ before $i_2$, then $i_1$ should be ordered no later than $i_2$ in the output.
Pomp\={e} adopts the fairness notion of \textit{removing oligarchy}, meaning that the output order cannot be dictated by a single leader.
In contrast, \sys enforces the \textit{equal opportunity} notion of fairness.

\paragraph{Setup and metrics.}
We run the three baselines and \sys on 12 machines in CloudLab~\cite{cloudlab} (ds430, Intel Xeon E5-2630, 64GB memory, Ubuntu Linux 24.04 LTS).
Since Intel Xeon E5-2630 does not support SGX,
we run the TEE-based SRO on a separate machine with the Intel Xeon Silver 4410Y processor.
We provide detailed instructions on running the experiments in an open-source artifact~\cite{cob}.

Statistics show that the top countries running Ethereum nodes are:
US (33\%),
Germany (15\%),
France (6\%),
Finland (6\%),
Singapore (4\%),
UK (4\%),
Canada (3\%),
Japan (3\%),
Netherlands (3\%), and
Australia (3\%) ~\cite{ethernodes}.
We thus start with a setup of 80 nodes where each node represents 1\% in the statistics.
For the US, we assume 11 nodes in Austin, 11 in San Francisco, and 11 in Washington, representing the central, western, and eastern US.
For the other countries, we assume that all the nodes are in one of their major cities.
There are thus 12 cities in this setup,
and we map these 12 cities to the 12 machines in CloudLab, emulating the latency between different cities with the Linux traffic control (\texttt{tc}) utility.
The latency information is from WonderNetwork~\cite{wondernetwork}.
We assume $\Delta_{net}=400ms$ in the partial synchrony model:
$400ms$ is used in Solana's latest consensus protocol~\cite{solanaalpenglow}, and $400ms$ is higher than the maximum latency in this setup ({\em i.e.,} $296ms$ from Canberra in Australia to Oulu in Finland).

We then introduce proof-of-stake, which reduces resource contention on the machines.
Specifically, instead of running 11 instances of the server-side software on one machine ({\em e.g.,} the machine mapped to Washington),
we run only one instance with $11/80$ of the total stake.
As explained in Section~\ref{s:impl},
this is essentially the same as the standard SMR model with 80 nodes that we just described.
Moreover, similar to real-world blockchains like Solana, we assign slot leaders based on their stake.
For example, in HotStuff, this machine will be the leader for $11/80$ of all the consensus slots.

Given a set of commands, the key metric for fairness is the difference between the probabilities of 
two permutations in the output.
For example, given two commands invoked at the same time in Washington and London,
the probability of the London one appearing first in the output is
0.475 (denoted as $Pr[L\prec W]=0.475$) in our measurement of HotStuff.
Therefore, $Pr[L\prec W] - Pr[W\prec L] = 0.475-0.525 = -0.05$ is shown in the first row of Figure~\ref{fig:geo-bias-table-1}.
Since $|Pr[L\prec W] - 1/2|=0.025$ is lower than $0.05$,
HotStuff treats the two clients from London and Washington equally if we target $\epsilon(2)=0.05$.
We will show that HotStuff is biased when considering clients from some other cities.

To answer the question about end-to-end performance, 
we measure the latency and throughput of HotStuff, Pomp\={e} and \sys in the 12-city setup.
We assume that clients are evenly distributed across the 12 cities,
and we increase the number of clients until the system is saturated.
Every client invokes commands in a closed loop ({\em i.e.,} it waits for enough responses of its currently outstanding command before invoking the next one),
and before saturation, the throughput increases as the number of clients increases.
We aim to measure the overhead of \sys over Pomp\={e} and compare their performance with HotStuff.
As for batching, we use $2000ms$ as the time interval associated with each consensus slot in Pomp\={e} and \sys.
We disabled batching in HotStuff and in the ordering phase of Pomp\={e} and \sys.

\subsection{Bias and front-running}
\label{s:eval:bias}

Figure~\ref{fig:geo-bias-table-1} shows that geographical bias could be significant in the baselines.
The output order produced by Pomp\={e} and Themis is deterministic.
When commands are invoked at the same time,
Pomp\={e} always produces London $\prec$ Munich $\prec$ Washington $\prec$ Tokyo,
reflecting the order of the median timestamps obtained by clients from these cities.
Themis produces Washington $\prec$ London because more than half (42/80) of the total stake receives the Washington command before a simultaneous one from London.
Since 42/80 is close to one half, the $\gamma$-batch-order-fairness property of Themis does not require Themis to order the Washington command first.
In contrast, 62/80 of the total stake receives the London command before a simultaneous one from Munich,
and $\gamma$-batch-order-fairness requires Themis to be biased towards London because 62/80 is higher than a threshold derived from the $\gamma$ parameter.
While $\gamma$-batch-order-fairness sometimes requires the system to be biased,
the notion of removing oligarchy does not require bias despite the fact that Pomp\={e} is significantly biased.

HotStuff rotates leadership, and the invocation from Tokyo would have a chance to appear first in the ledger when a Tokyo node serves as the leader.
However, this chance is still low compared to Munich invocations.
Specifically, we find that $Pr[\text{M} \prec \text{T}] = 0.74$, 
leading to the $0.74-0.26=0.48$ in the last row of Figure~\ref{fig:geo-bias-table-1}.
This difference increases to $0.83$ when compared to Washington.
While Tokyo and Washington are geographically distant, bias can happen between nearby cities.
London and Munich are both in Europe, but HotStuff is still biased towards London, as shown in the third row.

\begin{figure}[]
\fontsize{10pt}{12pt}\selectfont
\begin{tabularx}{1\columnwidth}{|l|l|l|l|}
\hline
\textbf{Baselines} & \textbf{HotStuff} & \textbf{Pomp\={e}} & \textbf{Themis}
\\ \hline
$Pr[L\prec W] - Pr[W\prec L]$
&  -0.05 $\checkmark$  & 1 & -1 \\
$Pr[W\prec T] - Pr[T\prec W]$
&  0.83 & 1 & 1 \\ 
$Pr[L\prec M] - Pr[M\prec L]$
&  0.48  & 1 & 1 \\
$Pr[M\prec T] - Pr[T\prec M]$
&  0.48 & 1 & 1 \\ \hline
\end{tabularx}

\vspace{-5px}
\caption{Geographical bias in HotStuff, Pomp\={e}, and Themis.
$W$, $L$, $M$, and $T$ stand for Washington, London, Munich, and Tokyo.
For two simultaneous invocations from two cities,
$Pr[\text{A} \prec \text{B}]$ stands for the probability of the invocation from city A being the first in the system output.
The checkmark for HotStuff means that bias is controlled by $\epsilon(2)=0.05$.
}
\label{fig:geo-bias-table-1}
\vspace{-10px}
\end{figure}

\begin{figure}[]
\fontsize{9.7pt}{12pt}\selectfont
\begin{tabularx}{1\columnwidth}{|l|l|l|}
\hline
\textbf{\sys} & $\Delta_{noise}=\Delta_{net}$ & $\Delta_{noise}=5*\Delta_{net}$
\\ \hline
$Pr[L\prec W] - Pr[W\prec L]$
&  0.24 & 0.05 $\checkmark$ \\
$Pr[W\prec T] - Pr[T\prec W]$
&  0.31 & 0.07 $\checkmark$ \\
$Pr[L\prec M] - Pr[M\prec L]$
&  0.07 $\checkmark$ & 0.01 $\checkmark$ \\
$Pr[M\prec T] - Pr[T\prec M]$
&  0.46 & 0.10 $\checkmark$ \\ \hline
\end{tabularx}

\vspace{-5px}
\caption{Geographical bias measured in \sys with $\Delta_{noise}=\Delta_{net}$ ({\em i.e.,} 400ms)
and $\Delta_{noise}=5*\Delta_{net}$ ({\em i.e.,} 2000ms). }
\label{fig:geo-bias-table-2}
\vspace{-5px}
\end{figure}

Figure~\ref{fig:geo-bias-table-2} shows how \sys reduces geographical bias.
After adding a random noise sampled uniformly from $[0, \Delta_{net}]$ ({\em i.e.,} $[0, 400ms]$) to the median timestamp of Pomp\={e},
clients from London and Munich will be treated equally.
With a target $\epsilon(2)=0.05$,
system designers could choose $\Delta_{noise}=5*\Delta_{net}$ ({\em i.e.,} $2000ms$), and the worst-case bias across these four cities will be effectively controlled.
Specifically, given the definition of $\epsilon$-Ordering Equality,
$\epsilon(2)=0.05$ means that $|Pr[M\prec T] - 1/2|$ should be at most $0.05$.
Therefore, the $0.10$ in Figure~\ref{fig:geo-bias-table-2} means that
$Pr[M\prec T]=0.55$ and $Pr[T\prec M]=0.45$,
satisfying the constraint put by $\epsilon(2)=0.05$.
In real-world deployments,
$\Delta_{noise}$ can be chosen by considering clients with the highest network latency.
\sys does not provide equal opportunity to clients with unbounded network latency to the system.
Indeed, even liveness cannot be guaranteed without a bound on network latency~\cite{fischer83impossibility},
and equal opportunity is meaningless without liveness.

Front-running typically occurs when one client has lower network latencies to a majority of the nodes than the other clients, as London or Washington do in this experiment.
By adding randomness,
\sys could give clients
more even chances of obtaining the \$89M profit in the liquidation events (explained in Section~\ref{s:violate-ordering-equality}).
Equal opportunity also plays an important role in mitigating sandwich attacks,
which is closely related to front-running.

\subsection{Sandwich attacks}
\label{s:eval:sandwich}

To explain sandwich attacks, we start with some background on decentralized exchange.
An exchange maintains a pool of some token A ({\em e.g.,} USD) and some token B ({\em e.g.,} CNY).
For example, people traveling from the US to China may put some USD into the pool and take some CNY away.
Similar pools exist on blockchains,
and the trading volume of Uniswap, a decentralized exchange application running on Ethereum, has exceeded one trillion dollars~\cite{uniswap}.

Pools in these exchanges follow a constraint: \textit{amount of tokenA} $*$ \textit{amount of tokenB} $=$ \textit{constant},
which is called the automated market maker (AMM) approach~\cite{amm}.
Suppose the constant is $1800$; the number of tokens A and B in the pool could be, for instance, $\langle 45, 40\rangle$, $\langle 60, 30\rangle$, or $\langle 75, 24\rangle$.
Say $\langle 75, 24\rangle$ is the current status, and Alice needs $15$ token A.
Alice can put $6$ token B into the pool and take $15$ token A away so that the pool state moves to $\langle 60, 30\rangle$.

The sandwich attack works as follows.
After seeing Alice's transaction, an attacker, Bob, first buys $15$ token A, moving the pool status to $\langle 60, 30\rangle$.
Alice now needs to pay $10$ (instead of $6$) token B to exchange for $15$ token A, moving the pool status to $\langle 45, 40\rangle$.
Bob can then exchange $15$ token A back to $10$ token B, making a profit of $4$ token B.
The three steps reflect the three invocations in Figure~\ref{fig:sandwich}.
If the market value of the tokens are $\$100$ and $\$200$, respectively, we will get the dollar values in Figure~\ref{fig:sandwich}.
Specifically, if the attack succeeds,
the victim will pay 10 token B ({\em i.e.,} $\$2000$) in exchange for 15 token A ({\em i.e.,} $\$1500$), losing $\$500$ as shown in the first row of Figure~\ref{fig:sandwich}.
Without the attack, the victim will pay 6 token B in exchange for 15 token A, making a profit of $\$300$.

A common way of conducting such sandwich attacks is to use faster network connections~\cite{lewis14flashboys}.
The attacker is typically co-located with the victim and,
after observing a trading command from the victim ($i_1$ in Figure~\ref{fig:sandwich}),
the attacker uses faster network connections to send its own command ($i_2$ in Figure~\ref{fig:sandwich})
to the nodes.
Figure~\ref{fig:sandwich-eval} shows the expected financial rewards in different systems.

In this experiment, we assume that an attacker co-located with the victim could deliver its command to all nodes in the system $10ms$ earlier than the victim.
The attack deterministically succeeds in all state-of-the-art systems because they all follow the order in which nodes receive commands.
In contrast, all 6 permutations are equally likely in \sys, lowering the attacker's expected reward  by an order of magnitude.
Specifically, $\$233=(-\$500+\$700+\$300*4)/6$ is the victim's expected reward,
and $\$67=(\$800-\$400)/6$ is the attacker's expected reward.
Given a much lower expected reward and a caveat of losing $\$400$,
\sys effectively disincentivizes sandwich attacks.

\begin{figure}[]
  \fontsize{11pt}{13pt}\selectfont
\begin{tabularx}{1\columnwidth}{|l|l|l|}
\hline
\textbf{Financial impact} & \textbf{Victim} & \textbf{Attacker}
\\ \hline
HotStuff / Pomp\={e} / Themis
&  -\$500
&  \$800 
\\ \hline
\sys
&  \$233
&  \$67
\\ \hline
\end{tabularx}
\caption{The expected financial rewards to the attacker and victim within the example sandwich attack shown in Figure~\ref{fig:sandwich}.}
\label{fig:sandwich-eval}
\vspace{-5px}
\end{figure}

\begin{figure*}[]
  \centering
  \includegraphics[width=1\textwidth]{fig/osdi26_combine.jpg}
  \caption{End-to-end performance of HotStuff, Pomp\={e}, and \sys with $\Delta_{noise}=5*\Delta_{net}$ ({\em i.e.,} $2000ms$) in the 12-city setup.
  Batching is disabled in HotStuff and in the ordering phase of Pomp\={e} or \sys.
  For example, HotStuff handles 3.6 consensus slots per second, and if we enable a batch of 1000 commands for each slot, the throughput of HotStuff will become 3600 commands per second.
  Disabling batching is useful because, in practice, protocols only reach consensus on the slot hashes.}
  \label{fig:scalability}
  \vspace{-10px}
\end{figure*}

\subsection{Latency of the two SRO designs}
\label{s:eval:sro}

\begin{figure}[]
\fontsize{10pt}{12pt}\selectfont
\begin{tabularx}{1\columnwidth}{|*{3}{X|}}
\hline
\begin{tabular}[c]{@{}l@{}}\end{tabular} & \begin{tabular}[c]{@{}l@{}}|\texttt{s}|=0 (\texttt{base})\end{tabular} & \begin{tabular}[c]{@{}l@{}}|\texttt{s}|=200\end{tabular}
\\ \hline
TEE
& 3us
& \texttt{base}+20.2ms \\ \hline
TVRF (67/100)
& 416+6346us
& \texttt{base}+19.9ms \\ \hline
TVRF (133/200)
& 394+12330us
& \texttt{base}+19.9ms \\ \hline
\end{tabularx}

\caption{Latency of the \texttt{Reveal} interface in different SRO implementations.
|\texttt{s}| denotes the number of signatures to be verified in the second parameter of \texttt{Reveal}.
The numbers in the parentheses are the threshold and total number of nodes for TVRF.
The \texttt{base} case of TVRF consists of two latency results for generating and combining shares.
}
\label{fig:sro-latency}
\vspace{-10px}
\end{figure}

Figure~\ref{fig:sro-latency} shows the latency of the two SRO designs. 
We show two cases of |\texttt{s}|=0 and |\texttt{s}|=200,
where |\texttt{s}| denotes the number of signatures to be verified in the second parameter of \texttt{Reveal}.
When |\texttt{s}|=0, \texttt{Reveal} does not verify signatures, and the latency is solely for generating random numbers.
For the TEE variant, the latency consists of entering an SGX enclave and computing a \texttt{SHA256} function, which takes only $3us$.
For the threshold VRF variant, the latency consists of three parts: (1) generating shares, (2) collecting the shares over the network,
(3) combining a threshold of shares.
The results in Figure~\ref{fig:sro-latency} show (1) and (3).
Under a setup of 100 nodes and 67 as the threshold, the latency of generating a share is $0.4ms$, and
the latency of combining 67 shares is $6.3ms$.
When moving to a setup of 200 nodes with 133 as the threshold,
the latency of generating a share stays at $0.4ms$,
while the latency of combining the shares doubles to $12ms$.
Lastly, when $|s|=200$, the latency of verifying the 200 signatures is about $20ms$, both within and outside an SGX enclave.

These results show a trade-off between performance and decentralization.
The SRO based on SGX has a much lower latency, making it more practical,
but it requires trusting Intel, a centralized party.
In the following experiments,
we choose performance in this trade-off and use the TEE variant of SRO,
since TEE has been used in Ethereum's Sepolia testnet~\cite{thirdteeblog}
and Solana's Block Assembly Marketplace~\cite{solanateeordering}.

\subsection{End-to-end performance of \sys}
\label{s:eval:end-to-end}

Figure~\ref{fig:scalability} shows the end-to-end performance of HotStuff, Pomp\={e}, and \sys.
We made three observations.

\vspace{-5px}
\begin{itemize}
    \item \sys keeps the same throughput as Pomp\={e}, and the latency increases by 1.12x-1.42x.
    \item Pomp\={e} and \sys achieve a significantly higher throughput than HotStuff in the geo-distributed setup.
    \item The 99\% latency of HotStuff is $287$-$592ms$ higher than its 50\% latency,
    while the 99\% latency of \sys is $1672$-$1816ms$ higher than its 50\% latency.
\end{itemize}
\vspace{-5px}

When deploying 120 clients in each city ({\em i.e.,} $120*12=1440$ clients in total),
Pomp\={e} and \sys reach their peak throughput: 1842 cmd/s for Pomp\={e} and 1893 cmd/s for \sys.
When considering clients from London,
the 50\% latency in \sys is \textit{1.29}x of the 50\% latency in Pomp\={e},
and the 99\% latency is \textit{1.42}x of the 99\% latency in Pomp\={e}.
For Canberra clients, the increase ratios are \textit{1.31}x and \textit{1.12}x for the 50\% and 99\% latencies, respectively.
This is how we derive the last statement in Figure~\ref{fig:eval-summary}.

HotStuff's performance issue in a geo-distributed setup has already been reported in prior work~\cite{zhang2020byzantine}.
Specifically, with rotating leadership in HotStuff, the leader of the current slot needs to wait for enough votes on the previous slot before proposing for the current slot.
Such waiting results in a low frequency of consensus proposals in HotStuff.
Pomp\={e} and \sys do not suffer from such waiting, although they reuse HotStuff for consensus.
Specifically, their leaderless ordering phase proposes commands continuously without waiting, and consensus only happens every $2000ms$ for all commands proposed by the leaderless ordering phase in a $2000ms$ time window.

Lastly, the gap between the 50\% and 99\% latencies in \sys is obviously larger than that gap in HotStuff,
and two factors cause such a larger gap.
First, since consensus happens every $2000ms$, there is a delay between the time when a client invokes a command and the time when consensus happens.
Second, due to random noise, it takes longer for a command to become stable in \sys, as explained in Section~\ref{s:design-integration}.
Both factors enlarge the gap between the 50\% and 99\% latencies in \sys.

\section{Related work}
\label{s:related}

\paragraph{Traditional BFT SMR systems.}
There is a long line of work on traditional BFT SMR
systems, including \cite{castro02practical,kotla07zyzzyva,haeberlen07peerreview,clement09upright,gueta18sbft,kotla04high,clement09making,hendricks10zzyzx,castro03base,li07beyond,yin03separating,martin06fast,porto15visigoth,liu16xft,behl17hybrids,
sousa2018byzantine,liu19scalable,suri2021basil,liu23flexible}.
These works focus on enforcing safety and liveness, removing or
constraining Byzantine influence, and improving performance or
theoretical complexity, among other goals.  Our work focuses on how
to choose the output order in SMR, which is not considered by the
traditional specification.  Some works
~\cite{chun07attested,levin09trinc,kapitza12cheapbft,behl17hybrids} use
trusted hardware to increase the ratio of Byzantine nodes that the
system can tolerate.  Other works
~\cite{gilad17scaling,christian2000Random,galindo2021fully} use
randomness to elect a committee or achieve safety and liveness in
a fully asynchronous model.  Unlike these works, our work uses
trusted hardware to provide a fault-tolerant source of randomness
and applies it to ensure equal opportunity.

\paragraph{Rotating leadership.}
Some works adopt a leader and rotate the leadership frequently, but they focus on reducing theoretical complexity or preventing faulty leaders from degrading performance.
Aardvark~\cite{clement09making} employs periodic leader changes to ensure a certain degree of performance in the presence of faulty leaders.
Aardvark sets an expectation on the minimal throughput a leader must ensure and triggers a leader change if the current leader fails to meet this expectation.
Unlike Aardvark,
HotStuff~\cite{yin19hotstuff} employs leader rotation and optimizes the communication complexity.
Specifically, HotStuff's communication complexity is linear in the number of nodes, which makes it more suitable for blockchains.
Adopted by Diem~\cite{diem},
rotating leadership in HotStuff aims to provide some sense of fairness in a permissioned blockchain.
Our work instead specifies and enforces a concrete notion of fairness, and our evaluation results show that rotating leadership could cause significant bias in a real-world deployment.

\paragraph{Removing oligarchy.}
{\em Leaderless} protocols argue against having a leader node who can unilaterally decide the output order.
EPaxos~\cite{moraru13more}, or Egalitarian Paxos, is an SMR protocol that attempts to make the system egalitarian.
While the concept of egalitarianism is closely related to equal opportunity,
EPaxos does not specify or enforce egalitarian ideals except being a leaderless protocol.
Byzantine oligarchy~\cite{zhang2020byzantine} is the first attempt to specify the goal of leaderless protocols in the context of ordering, and
Pomp\={e} is the first leaderless protocol 
that provably removes a Byzantine oligarchy. 
To achieve this, 
Pomp\={e} requires a
client to collect timestamps from a quorum of nodes, and the median is used to order a command.
However, our evaluation shows that using such a median timestamp could make the system even more biased
than prior approaches such as HotStuff.

\paragraph{Decentralized first-come-first-served.}
Some recent works define and enforce fairness concepts related to first-come-first-served~\cite{kelkar20order, kelkar2022order,kelkar2021themis,cachin2022quick,ramseyer2023fair,ke2024separation}.
Specifically,
these protocols enforce variants of the receive-order-fairness property~\cite{kelkar20order}, which essentially says that if a majority of the nodes receive an invocation first, it should be ordered first in the output.
We argue that this property can be unfair because, without distinguishing relevant features from irrelevant ones, it can amplify systemic bias in real-world blockchains, as shown in the evaluation section.

While these works enforce specific properties related to first-come-first-served,
the framework of ordered consensus makes it possible to prove that it is impossible, in general, to prevent Byzantine replicas from manipulating the order ({\em e.g.,} from conducting front-running)~\cite{zhang2020byzantine}.
This result is inspired by
Arrow’s impossibility theorem~\cite{arrow2012social}
and the Gibbard-Satterthwaite impossibility theorem~\cite{gibbard1973manipulation, satterthwaite1975strategy}
from the field of social choice theory.
In the past two decades, computer scientists became interested in social choice, leading to the creation of the field of computational social choice~\cite{brandt2016handbook}.

\paragraph{Game theory.}
The BAR model\cite{abraham2011distributed,aiyer05bar} explores how to
connect Byzantine fault tolerance to game theory. The core of the
connection is adopting the Nash theorem~\cite{nash1951non}, which
states that every normal-form game must have an equilibrium. The Nash
theorem connects the rationality model~\cite{von1947theory,savage1972foundations} with the Brouwer fixed-point theorem ({\em i.e.,} a
fixed point is interpreted as an equilibrium).
We find the concept of equal opportunity and its violations more
prevalent in real-life scenarios than the Nash equilibrium. We reuse
the key concept of expected value from the rationality model when
explaining violations of equal opportunity in
Section~\ref{s:violate-ordering-equality}. We are also inspired by
one of the axioms at the core of the rationality
model~\cite{savage1972foundations} when defining consistency in
Section~\ref{s:two-principles}. However, unlike BAR, we do not reuse
the concept of Nash equilibrium in this work.

Another line of work studies how to eliminate sandwich attacks in blockchains within a game-theoretic model~\cite{alpos2023eatingsandwichesmodularlightweight}.
Similar to our work, this work also argues that randomness is the key to addressing the problem.
The proposed solution only randomizes commands within a block, making it vulnerable to destructive front-running attacks described in an empirical study~\cite{qin2022quantifying}.
Destructive front-running means that, in leader-based protocols, the leader of a slot could front-run a command by excluding it from this slot, so this victim command can only be included in a later slot.

\paragraph{Equal opportunity in real-life scenarios.}
The principles of impartiality and consistency have been scrutinized
in the context of how society allocates resources. In his
book~\cite{young1995equity}, Young studies the two principles in a
variety of contexts, from employment to kidney exchange, and discusses
how they are embodied in key pieces of legislation
~\cite{civilrightact,originalempolymentact,organact,organactnew}. The
book includes a proof that a point system is the only mechanism that
satisfies both impartiality and consistency. Our work adopts this book's approach: we use the invocation time as the score for ordering and
use a Secret Random Oracle to break ties.

Kidneys for transplant used to be exchanged through a free market, and
wealthy patients had a better chance of getting kidneys. This raised
fairness concerns and led to legislation that transferred the
operation to the government to overcome such bias towards
wealthy patients~\cite{organact}. While this law enforced
\textit{impartiality} between the rich and poor, \textit{consistency}
became a concern.  The first algorithm proposed was inconsistent, and
the order of two patients getting kidneys could be switched due to a
third patient joining the system. As a result, a new algorithm was
proposed as an amendment that enforces the consistency
principle~\cite{organactnew}. Details of the relevant features and
the point system in this case have been discussed in Chapter 2 of Young's
book~\cite{young1995equity}.

Financial regulation laws require financial exchanges to be impartial
to all traders. A recent expos\'e~\cite{lewis14flashboys}, however,
has concluded that high-frequency traders have routinely engaged in
market-exploiting behaviors, such as aggressive latency optimizations and front-running by exploiting their location or the availability of fast network connections---which our framework would classify as irrelevant features.

In the Olympic Games, the order of athletes is decided, and equal opportunity is required.
The Olympic rules specify which relevant features should be measured.
Similar to geolocation, a key irrelevant feature is national origin,
and judgments should not be biased toward any nation.
Besides, due to cognitive bias, a judge may make correlated decisions when judging a sequence of observations.
For example, in diving, athletes dive alternately, and an athlete's performance may impact the scores given to the next athlete~\cite{kramer2017sequential}.
Judges typically need professional training to
combat such implicit bias and follow the consistency principle.

\section{Conclusion}

This paper introduces a model for equal opportunity, a notion of fairness based on the distinction between relevant and irrelevant features. 
Existing protocols---including ones attempting to provide some fairness---can be significantly biased and vulnerable to ordering-based attacks.

We design, implement, and evaluate \sys{}, a new ordered consensus protocol that guarantees two correctness conditions for equal opportunity, $\epsilon$-Ordering Equality and $\Delta$-Ordering Separation. 
\sys{} effectively mitigates the well-known ordering-based attacks,
given that eliminating such attacks in the presence of Byzantine
influence has been proved impossible by prior work.

\subsubsection*{Acknowledgments}

We thank the anonymous shepherd and OSDI reviewers for their insightful and constructive comments.
The initial steps towards a theory model for equal opportunity benefited from conversations with Joseph Halpern and his book on theories of uncertainty~\cite{halpern2017reasoning}.
Roger Wattenhofer helped us confirm our understanding of Solana's new consensus protocol.
This work was supported in part by a Facebook Fellowship.

{
\begin{flushleft}
\bibliographystyle{plain}
\bibliography{conferences,refs}

\begin{thebibliography}{10}

\bibitem{civilrightact}
{Civil Rights Act}.
\newblock \url{https://www.congress.gov/bill/88th-congress/house-bill/7152},
  1964.

\bibitem{originalempolymentact}
{Equal Employment Opportunity Act}.
\newblock
  \url{https://www.eeoc.gov/history/equal-employment-opportunity-act-1972},
  1972.

\bibitem{empolymentact}
{Uniform guidelines on employee selection procedures}.
\newblock
  \url{https://www.ecfr.gov/current/title-29/subtitle-B/chapter-XIV/part-1607},
  1978.

\bibitem{organact}
{National Organ Transplant Act}.
\newblock \url{https://www.congress.gov/bill/98th-congress/senate-bill/2048},
  1984.

\bibitem{organactnew}
{National Organ Transplant Program Extension Act}.
\newblock \url{https://www.congress.gov/bill/101st-congress/house-bill/5146},
  1990.

\bibitem{makerdao}
{MakerDAO}.
\newblock \url{https://makerdao.com}, 2017.

\bibitem{uniswap}
{Uniswap}.
\newblock \url{https://uniswap.org/}, 2018.

\bibitem{diem}
{Diem}.
\newblock \url{https://www.diem.com/en-us/}, 2019.

\bibitem{aave}
{Aave liquidity protocol}.
\newblock \url{https://aave.com/}, 2020.

\bibitem{dvrf}
A threshold {VRF} implementation.
\newblock \url{https://github.com/fetchai/research-dvrf}, 2020.

\bibitem{ethmerge}
{The Merge in Ethereum}.
\newblock \url{https://ethereum.org/roadmap/merge/}, 2022.

\bibitem{thirdteeblog}
Block building inside {SGX}.
\newblock \url{https://writings.flashbots.net/block-building-inside-sgx}, 2023.

\bibitem{cloudlab}
{CloudLab}.
\newblock \url{https://cloudlab.us/}, 2023.

\bibitem{ethernodes}
Ethereum mainnet statistics.
\newblock \url{https://ethernodes.org/countries?synced=1}, 2023.

\bibitem{wondernetwork}
Global ping statistics: Ping times between {WonderNetwork} servers.
\newblock \url{https://wondernetwork.com/pings}, 2023.

\bibitem{sgxlinux}
Intel software guard extensions for {Linux OS}.
\newblock \url{https://github.com/intel/linux-sgx}, 2023.

\bibitem{sgx}
{Intel Software Guard Extensions ({SGX})}.
\newblock
  \url{https://www.intel.com/content/www/us/en/architecture-and-technology/software-guard-extensions.html},
  2023.

\bibitem{mcl}
mcl: A portable and fast pairing-based cryptography library.
\newblock \url{https://github.com/herumi/mcl}, 2023.

\bibitem{secp256k1}
The secp256k1 library from {Bitcoin}.
\newblock \url{https://github.com/bitcoin-core/secp256k1}, 2023.

\bibitem{amm}
What are {Automated Market Makers} ({AMMs})?
\newblock
  \url{https://chain.link/education-hub/what-is-an-automated-market-maker-amm},
  2023.

\bibitem{solanateeordering}
Solana's {Block Assembly Marketplace}.
\newblock \url{https://bam.dev/}, 2025.

\bibitem{cob}
Benchmark for equal opportunity in ordered consensus.
\newblock \url{https://github.com/yhzhang0128/cob}, 2026.

\bibitem{abraham2011distributed}
Ittai Abraham, Lorenzo Alvisi, and Joseph~Y. Halpern.
\newblock Distributed computing meets game theory: {Combining} insights from
  two fields.
\newblock {\em ACM Sigact News}, 42(2):69--76, 2011.

\bibitem{aiyer05bar}
Amitanand~S. Aiyer, Lorenzo Alvisi, Allen Clement, Mike Dahlin, Jean-Philippe
  Martin, and Carl Porth.
\newblock {BAR fault tolerance for cooperative services}.
\newblock In {\em Proceedings of the ACM Symposium on Operating Systems
  Principles (SOSP)}, pages 45--58, 2005.

\bibitem{alpos2023eatingsandwichesmodularlightweight}
Orestis Alpos, Ignacio Amores-Sesar, Christian Cachin, and Michelle Yeo.
\newblock Eating sandwiches: Modular and lightweight elimination of transaction
  reordering attacks.
\newblock arXiv:2307.02954, 2023.

\bibitem{arrow2012social}
Kenneth~J. Arrow.
\newblock {\em Social choice and individual values}, volume~12.
\newblock Yale University Press, 1951.

\bibitem{barocas-hardt-narayanan}
Solon Barocas, Moritz Hardt, and Arvind Narayanan.
\newblock {\em Fairness and Machine Learning: Limitations and Opportunities}.
\newblock MIT Press, 2023.

\bibitem{behl17hybrids}
Johannes Behl, Tobias Distler, and Rudiger Kapitza.
\newblock Hybrids on steroids: {SGX}-based high performance {BFT}.
\newblock In {\em Proceedings of the ACM European Conference on Computer
  Systems (EuroSys)}, pages 222--237, 2017.

\bibitem{brandt2016handbook}
Felix Brandt, Vincent Conitzer, Ulle Endriss, J{\'e}r{\^o}me Lang, and Ariel~D.
  Procaccia.
\newblock {\em Handbook of computational social choice}.
\newblock Cambridge University Press, 2016.

\bibitem{christian2000Random}
Christian Cachin, Klaus Kursawe, and Victor Shoup.
\newblock Random oracles in {Constantinople}: Practical asynchronous
  {Byzantine} agreement using cryptography.
\newblock In {\em Proceedings of the ACM Symposium on Principles of Distributed
  Computing (PODC)}, page 123–132, 2000.

\bibitem{cachin2022quick}
Christian Cachin, Jovana Mićić, Nathalie Steinhauer, and Luca Zanolini.
\newblock {Quick Order Fairness}.
\newblock arXiv:2112.06615, 2022.

\bibitem{castro02practical}
Miguel Castro and Barbara Liskov.
\newblock Practical {Byzantine} fault tolerance and proactive recovery.
\newblock {\em ACM Transactions on Computer Systems (TOCS)}, 20(4):398--461,
  2002.

\bibitem{castro03base}
Miguel Castro, Rodrigo Rodrigues, and Barbara Liskov.
\newblock {BASE: Using} abstraction to improve fault tolerance.
\newblock {\em ACM Transactions on Computer Systems (TOCS)}, 21(3):236--269,
  2003.

\bibitem{chun07attested}
Byung-Gon Chun, Petros Maniatis, Scott Shenker, and John Kubiatowicz.
\newblock {Attested Append-only Memory}: Making adversaries stick to their
  word.
\newblock In {\em Proceedings of the ACM Symposium on Operating Systems
  Principles (SOSP)}, pages 189--204, 2007.

\bibitem{clement09upright}
Allen Clement, Manos Kapritsos, Sangmin Lee, Yang Wang, Lorenzo Alvisi, Mike
  Dahlin, and Taylor Riche.
\newblock {UpRight} cluster services.
\newblock In {\em Proceedings of the ACM Symposium on Operating Systems
  Principles (SOSP)}, pages 277--290, 2009.

\bibitem{clement09making}
Allen Clement, Edmund Wong, Lorenzo Alvisi, Mike Dahlin, and Mirco Marchetti.
\newblock {Making Byzantine fault tolerant systems tolerate Byzantine faults}.
\newblock In {\em Proceedings of the USENIX Symposium on Networked Systems
  Design and Implementation (NSDI)}, pages 153--168, 2009.

\bibitem{daian20flash}
Philip Daian, Steven Goldfeder, Tyler Kell, Yunqi Li, Xueyuan Zhao, Iddo
  Bentov, Lorenz Breidenbach, and Ari Juels.
\newblock {Flash Boys 2.0}: Frontrunning, transaction reordering, and consensus
  instability in decentralized exchanges.
\newblock In {\em Proceedings of the IEEE Symposium on Security and Privacy
  (S\&P)}, pages 910--927, 2020.

\bibitem{dwork88consensus}
Cynthia Dwork, Nancy Lynch, and Larry Stockmeyer.
\newblock Consensus in the presence of partial synchrony.
\newblock {\em Journal of the ACM (JACM)}, 35(2), 1988.

\bibitem{fischer83impossibility}
Michael~J. Fischer, Nancy~A. Lynch, and Michael~S. Paterson.
\newblock Impossibility of distributed consensus with one faulty process.
\newblock In {\em Proceedings of the Symposium on Principles of Database
  Systems}, pages 1--7, 1983.

\bibitem{galindo2021fully}
David Galindo, Jia Liu, Mihair Ordean, and Jin-Mann Wong.
\newblock Fully distributed verifiable random functions and their application
  to decentralised random beacons.
\newblock In {\em Proceedings of the IEEE European Symposium on Security and
  Privacy (EuroS\&P)}, pages 88--102, 2021.

\bibitem{gibbard1973manipulation}
Allan Gibbard.
\newblock Manipulation of voting schemes: a general result.
\newblock {\em Econometrica: Journal of the Econometric Society}, pages
  587--601, 1973.

\bibitem{gilad17scaling}
Yossi Gilad, Rotem Hemo, Silvio~M. Micali, Georgios Vlachos, and Nickolai
  Zeldovich.
\newblock Algorand: Scaling {Byzantine} agreements for cryptocurrencies.
\newblock In {\em Proceedings of the ACM Symposium on Operating Systems
  Principles (SOSP)}, pages 51--68, 2017.

\bibitem{gueta18sbft}
Guy~Golan Gueta, Ittai Abraham, Shelly Grossman, Dahlia Malkhi, Benny Pinkas,
  Michael~K. Reiter, Dragos-Adrian Seredinschi, Orr Tamir, and Alin Tomescu.
\newblock {SBFT}: A scalable decentralized trust infrastructure for
  blockchains.
\newblock arxiv:1804/01626v1, 2018.

\bibitem{haeberlen07peerreview}
Andreas Haeberlen, Petr Kouznetsov, and Peter Druschel.
\newblock {PeerReview:} {Practical} accountability for distributed systems.
\newblock In {\em Proceedings of the ACM Symposium on Operating Systems
  Principles (SOSP)}, pages 175--188, 2007.

\bibitem{halpern2017reasoning}
Joseph~Y. Halpern.
\newblock {\em Reasoning about Uncertainty}.
\newblock MIT Press, 2017.

\bibitem{hendricks10zzyzx}
James Hendricks, Shafeeq Sinnamohideen, Gregory~R. Ganger, and Michael~K.
  Reiter.
\newblock {Zzyzx: Scalable fault tolerance through {Byzantine} locking}.
\newblock In {\em Proceedings of the International Conference on Dependable
  Systems and Networks (DSN)}, pages 363--372, 2010.

\bibitem{kapitza12cheapbft}
R\"{u}diger Kapitza, Johannes Behl, Christian Cachin, Tobias Distler, Simon
  Kuhnle, Seyed~Vahid Mohammadi, Wolfgang Schr\"{o}der-Preikschat, and Klaus
  Stengel.
\newblock {CheapBFT}: Resource-efficient {Byzantine} fault tolerance.
\newblock In {\em Proceedings of the ACM European Conference on Computer
  Systems (EuroSys)}, pages 295--308, 2012.

\bibitem{kelkar2022order}
Mahimna Kelkar, Soubhik Deb, and Sreeram Kannan.
\newblock Order-fair consensus in the permissionless setting.
\newblock In {\em Proceedings of the ACM on ASIA Public-Key Cryptography
  Workshop}, pages 3--14, 2022.

\bibitem{kelkar2021themis}
Mahimna Kelkar, Soubhik Deb, Sishan Long, Ari Juels, and Sreeram Kannan.
\newblock Themis: Fast, strong order-fairness in {Byzantine} consensus.
\newblock In {\em Proceedings of the ACM Conference on Computer and
  Communications Security (CCS)}, page 475–489, 2023.

\bibitem{kelkar20order}
Mahimna Kelkar, Fan Zhang, Steven Goldfeder, and Ari Juels.
\newblock Order-fairness for {Byzantine} consensus.
\newblock In {\em Proceedings of the International Cryptology Conference
  (CRYPTO)}, pages 451--480, 2020.

\bibitem{solanaalpenglow}
Quentin Kniep, Jakub Sliwinski, and Roger Wattenhofer.
\newblock Solana {Alpenglow} consensus: Increased bandwidth, reduced latency.
\newblock Technical report, Anza, 2025.
\newblock \url{https://www.anza.xyz/alpenglow-1-1}.

\bibitem{kotla07zyzzyva}
Ramakrishna Kotla, Lorenzo Alvisi, Mike Dahlin, Allen Clement, and Edmund Wong.
\newblock Zyzzyva: Speculative {Byzantine} fault tolerance.
\newblock In {\em Proceedings of the ACM Symposium on Operating Systems
  Principles (SOSP)}, pages 45--58, 2007.

\bibitem{kotla04high}
Ramakrishna Kotla and Mike Dahlin.
\newblock {High throughput Byzantine fault tolerance}.
\newblock In {\em Proceedings of the International Conference on Dependable
  Systems and Networks (DSN)}, pages 575--584, 2004.

\bibitem{kramer2017sequential}
Robin~S.S. Kramer.
\newblock Sequential effects in {Olympic} synchronized diving scores.
\newblock {\em Royal Society Open Science}, 4(1), 2017.

\bibitem{lamport98parttime}
Leslie Lamport.
\newblock The part-time parliament.
\newblock {\em ACM Transactions on Computer Systems (TOCS)}, 16(2):133--169,
  1998.

\bibitem{levin09trinc}
Dave Levin, John~R. Douceur, Jacob~R. Lorch, and Thomas Moscibroda.
\newblock {TrInc}: Small trusted hardware for large distributed systems.
\newblock In {\em Proceedings of the USENIX Symposium on Networked Systems
  Design and Implementation (NSDI)}, pages 1--14, 2009.

\bibitem{lewis14flashboys}
Michael Lewis.
\newblock {\em {Flash Boys:} A wall street revolt}.
\newblock {W. W. Norton \& Company}, 2014.

\bibitem{li07beyond}
Jinyuan Li and David Mazi{\'e}res.
\newblock Beyond one-third faulty replicas in {Byzantine} fault tolerant
  systems.
\newblock In {\em Proceedings of the USENIX Symposium on Networked Systems
  Design and Implementation (NSDI)}, pages 131--144, 2007.

\bibitem{liu19scalable}
Jian {Liu}, Wenting {Li}, Ghassan~O. {Karame}, and N.~{Asokan}.
\newblock Scalable {Byzantine} consensus via hardware-assisted secret sharing.
\newblock {\em IEEE Transactions on Computers}, 68(1), 2019.

\bibitem{liu16xft}
Shengyun Liu, Paolo Viotti, Christian Cachin, Vivien Qu{\'e}ma, and Marko
  Vukolic.
\newblock {XFT:} {Practical} fault tolerance beyond crashes.
\newblock In {\em Proceedings of the USENIX Symposium on Operating Systems
  Design and Implementation (OSDI)}, pages 485--500, 2016.

\bibitem{liu23flexible}
Shengyun Liu, Wenbo Xu, Chen Shan, Xiaofeng Yan, Tianjing Xu, Bo~Wang, Lei Fan,
  Fuxi Deng, Ying Yan, and Hui Zhang.
\newblock Flexible advancement in asynchronous {BFT} consensus.
\newblock In {\em Proceedings of the ACM Symposium on Operating Systems
  Principles (SOSP)}, page 264–280, 2023.

\bibitem{martin06fast}
Jean-Philippe Martin and Lorenzo Alvisi.
\newblock Fast {Byzantine} consensus.
\newblock {\em IEEE Transactions on Dependable and Secure Computing},
  3(3):202--215, 2006.

\bibitem{moraru13more}
Iulian Moraru, David~G. Andersen, and Michael Kaminsky.
\newblock There is more consensus in egalitarian parliaments.
\newblock In {\em Proceedings of the ACM Symposium on Operating Systems
  Principles (SOSP)}, pages 358--372, 2013.

\bibitem{ke2024separation}
Ke~Mu, Bo~Yin, Alia Asheralieva, and Xuetao Wei.
\newblock Separation is good: A faster order-fairness {Byzantine} consensus.
\newblock In {\em Proceedings of the Network and Distributed System Security
  Symposium (NDSS)}, 2024.

\bibitem{nakamoto08bitcoin}
Satoshi Nakamoto.
\newblock Bitcoin: A peer-to-peer electronic cash system.
\newblock Whitepaper, October 2008.

\bibitem{nash1951non}
John Nash.
\newblock Non-cooperative games.
\newblock {\em Annals of Mathematics}, pages 286--295, 1951.

\bibitem{porto15visigoth}
Daniel Porto, Jo\~{a}o Leit\~{a}o, Cheng Li, Allen Clement, Aniket Kate, Flavio
  Junqueira, and Rodrigo Rodrigues.
\newblock Visigoth fault tolerance.
\newblock In {\em Proceedings of the ACM European Conference on Computer
  Systems (EuroSys)}, pages 1--14, 2015.

\bibitem{qin2022quantifying}
Kaihua Qin, Liyi Zhou, and Arthur Gervais.
\newblock Quantifying blockchain extractable value: How dark is the forest?
\newblock In {\em Proceedings of the IEEE Symposium on Security and Privacy
  (S\&P)}, pages 198--214, 2022.

\bibitem{ramseyer2023fair}
Geoffrey Ramseyer and Ashish Goel.
\newblock Fair ordering via streaming social choice theory.
\newblock arXiv:2304.02730, 2023.

\bibitem{satterthwaite1975strategy}
Mark~Allen Satterthwaite.
\newblock Strategy-proofness and arrow's conditions: Existence and
  correspondence theorems for voting procedures and social welfare functions.
\newblock {\em Journal of Economic Theory}, 10(2):187--217, 1975.

\bibitem{savage1972foundations}
Leonard~J. Savage.
\newblock {\em The Foundations of Statistics}.
\newblock Courier Corporation, 1972.

\bibitem{schneider90implementing}
Fred~B. Schneider.
\newblock Implementing fault-tolerant services using the state machine
  approach: A tutorial.
\newblock {\em ACM Computing Surveys}, 22(4):299--319, 1990.

\bibitem{sousa2018byzantine}
Joao Sousa, Alysson Bessani, and Marko Vukolic.
\newblock A {Byzantine} fault-tolerant ordering service for the {Hyperledger
  Fabric} blockchain platform.
\newblock In {\em Proceedings of the International Conference on Dependable
  Systems and Networks (DSN)}, pages 51--58, 2018.

\bibitem{suri2021basil}
Florian Suri-Payer, Matthew Burke, Zheng Wang, Yunhao Zhang, Lorenzo Alvisi,
  and Natacha Crooks.
\newblock Basil: Breaking up {BFT} with {ACID} (transactions).
\newblock In {\em Proceedings of the ACM Symposium on Operating Systems
  Principles (SOSP)}, pages 1--17, 2021.

\bibitem{torres2021frontrunner}
Christof~Ferreira Torres, Ramiro Camino, and Radu State.
\newblock Frontrunner jones and the raiders of the dark forest: An empirical
  study of frontrunning on the {Ethereum} blockchain.
\newblock In {\em Proceedings of the USENIX Security Symposium}, pages
  1343--1359, 2021.

\bibitem{von1947theory}
John von Neumann and Oskar Morgenstern.
\newblock {\em Theory of Games and Economic Behavior}.
\newblock Princeton University Press, 1947.

\bibitem{yin03separating}
Jian Yin, Jean-Philippe Martin, Arun Venkataramani, Lorenzo Alvisi, and Mike
  Dahlin.
\newblock {Separating agreement from execution for {Byzantine} fault tolerant
  services}.
\newblock In {\em Proceedings of the ACM Symposium on Operating Systems
  Principles (SOSP)}, pages 253--267, 2003.

\bibitem{yin19hotstuff}
Maofan Yin, Dahlia Malkhi, Michael~K. Reiter, Guy~Golan Gueta, and Ittai
  Abraham.
\newblock {HotStuff:} {BFT} consensus with linearity and responsiveness.
\newblock In {\em Proceedings of the ACM Symposium on Principles of Distributed
  Computing (PODC)}, pages 347--356, 2019.

\bibitem{young1995equity}
H.~Peyton Young.
\newblock {\em Equity: In Theory and Practice}.
\newblock Princeton University Press, 1995.

\bibitem{zhang2020byzantine}
Yunhao Zhang, Srinath Setty, Qi~Chen, Lidong Zhou, and Lorenzo Alvisi.
\newblock Byzantine ordered consensus without {Byzantine} oligarchy.
\newblock In {\em Proceedings of the USENIX Symposium on Operating Systems
  Design and Implementation (OSDI)}, pages 633--649, 2020.

\bibitem{zhou2021high}
Liyi Zhou, Kaihua Qin, Christof~Ferreira Torres, Duc~V. Le, and Arthur Gervais.
\newblock High-frequency trading on decentralized on-chain exchanges.
\newblock In {\em Proceedings of the IEEE Symposium on Security and Privacy
  (S\&P)}, pages 428--445, 2021.

\end{thebibliography}
\end{flushleft}
}

\ifx\buildnoproof\undefined
\newpage

\newcommand\Dnoise{{\Delta_{noise}}}
\newcommand\Dnet{{\Delta_{net}}}

\newcommand{\multiset}[2]{
\left.\mathchoice
  {\left(\kern-0.48em\binom{#1}{#2}\kern-0.48em\right)}
  {\big(\kern-0.30em\binom{\smash{#1}}{\smash{#2}}\kern-0.30em\big)}
  {\left(\kern-0.30em\binom{\smash{#1}}{\smash{#2}}\kern-0.30em\right)}
  {\left(\kern-0.30em\binom{\smash{#1}}{\smash{#2}}\kern-0.30em\right)}
\right.}

\theoremstyle{plain}
\newtheorem*{restatedinner}{\restatedname}
\newenvironment{restated}[1]
  {\def\restatedname{#1}\begin{restatedinner}}
  {\end{restatedinner}}

\begin{appendices}

\section{Proofs and further analysis}
\label{s:appendix}

In this appendix, statements that already appear in the main text are restated under their original numbers; new statements are numbered anew.

We have also formalized the analysis in the Lean~4 theorem prover with the Mathlib library.
Starting from the timing model of Section~\ref{s:appendix-setup}, Theorems~\ref{theorem-ordering-separation}--\ref{theorem-ordering-equality-n}, \ref{theorem-continuous-eps-2}, and~\ref{theorem-continuous-eps-n} are all machine-checked.
The formalization is available at \url{https://github.com/FTRobbin/equal-opportunity-lean}.

\subsection{Setup}
\label{s:appendix-setup}

We assume that every invocation $i_j$ is invoked at $T_j$ and is associated with an assigned timestamp $t_j$ by the ordered consensus protocol.
The final ordering of the commands is based on their modified timestamps, $t'_j$, derived by adding independently sampled noise, $\eta_j$, to the assigned timestamps.

The ordered consensus protocol needs to guarantee that every assigned timestamp $t_j \in [T_j, T_j + \Dnet)$.
For Pomp\={e}, this follows from the assumption and lemma of Section~\ref{s:design-quantify}, the latter proven in~\cite{zhang2020byzantine}.

\begin{restated}{Assumption~\ref{assumption-sync-net}}
After GST,
if a command is invoked at time $T$, correct nodes will provide timestamps in the range $[T, T + \Delta_{net})$ for this command.
\end{restated}

\begin{restated}{Lemma~\ref{lemma-assigned-timestamp}}
  The assigned timestamp of a command is bounded by timestamps provided by correct nodes.
\end{restated}

We assume a non-adaptive adversary model by default,
i.e., the adversary must commit its timestamp picks before the noise is revealed.
An ideal SRO would guarantee this.
The only exception is the tightness of Theorem~\ref{theorem-continuous-eps-n},
which only applies to an adaptive adversary.

\subsection{$\Delta$-Ordering Separation}

We recall the definition from Section~\ref{s:two-properties}.

\begin{definition}[$\Delta$-Ordering Separation]
For any invocations $i_1, i_2 \in \Io_{gt}$, if $i_1$ is invoked at least $\Delta$ time units before
$i_2$, then $i_1$ precedes $i_2$ in the system output.
\end{definition}

As proven in Section~\ref{s:design-quantify}, any distribution on the interval $[0, \Dnoise]$ satisfies $\Delta=(\Dnet + \Dnoise)$-Ordering Separation:

\begin{restated}{Theorem~\ref{theorem-ordering-separation}}($\Delta$-Ordering Separation)
After GST, for all invocations $i_1$ and $i_2$, if $i_1$ is invoked more than $\Delta_{net}+\Delta_{noise}$ time units earlier than $i_2$, then $i_1$ is guaranteed to appear before $i_2$ in the output.
\end{restated}

Thus, we focus on deriving bounds for $\epsilon$ assuming a fixed $\Dnoise$.

\subsection{$\epsilon$-Ordering Equality}

We recall the definition from Section~\ref{s:two-properties}.

\begin{definition}[$\epsilon$-Ordering Equality]
For any subset $\Io \subseteq \Io_{gt}$ in which all invocations share the same
invocation time, let $n$ denote the number of invocations in $\Io$.
For any total order $\succ$ over $\Io$,  $|Pr[\succ]- \frac{1}{n!}| \leq \epsilon(n)$.
\end{definition}

We first provide proofs for the discrete distribution described in Section~\ref{s:design-quantify},
followed by results for the uniform continuous distribution on $[0, \Dnoise]$ used in our evaluation.

\subsection{Discrete distribution}

\begin{definition}[Discrete distribution]
    \label{def:discrete}
    Let $k$ be the natural number such that $k\Dnet \le \Dnoise < (k + 1)\Dnet$.
    The discrete distribution is uniformly random across the $k + 1$ discrete points of $0, \Dnet, 2\Dnet, \dots, k\Dnet$.
\end{definition}

Here is a visualization of the distribution for intuition:
\begin{figure}[h]
    \centering
    \includegraphics[width=\linewidth]{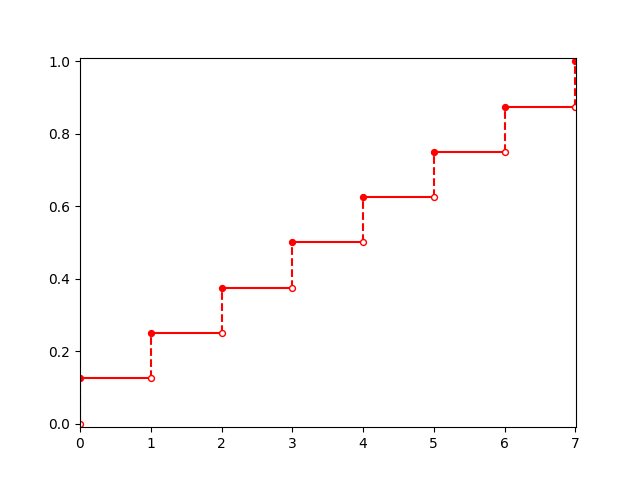}
    \caption{CMF of the discrete distribution that gives the optimal $\epsilon(2)$ when $\Dnoise = 7\Dnet$.}
    \label{fig:discrete-cmf}
\end{figure}

Section~\ref{s:design-quantify} already proves a lower bound on $\epsilon(2)$ for any distribution and shows that the discrete distribution achieves it:

\begin{restated}{Theorem~\ref{theorem-ordering-equality-2}}($\epsilon(2)$-Ordering Equality lower bound)
After GST,
for any probability distribution $D$ on $[0, \Delta_{noise}]$,
for all invocations $i_1$ and $i_2$ invoked at the same time,
$|Pr[i_1 \prec i_2 ] - \frac{1}{2}| \ge \frac{1}{2(k + 1)}$
where $k$ is the natural number such that $k\Delta_{net} \le \Delta_{noise} < (k + 1)\Delta_{net}$.
\end{restated}

\begin{restated}{Theorem~\ref{theorem-ordering-equality-3}}(Discrete distribution optimality)
	The above lower bound on $\epsilon(2)$ can be achieved by setting $D$ to be a discrete distribution
	that is uniformly random across the discrete points of $0, \Delta_{net}, 2\Delta_{net}, \dots, k\Delta_{net}$.
\end{restated}

That is, the discrete distribution defined in Definition \ref{def:discrete} achieves $\epsilon(2) = \frac{1}{2(k + 1)}$.

We now prove the more general theorem stated in Section~\ref{s:design-quantify}.

\begin{restated}{Theorem~\ref{theorem-ordering-equality-n}}($\epsilon(n)$-Ordering Equality)
After GST,
for all invocations $i_1..i_n$ invoked at the same time,
for the discrete distribution constructed as above,
for any total order $\succ$ of $i_1..i_n$,
$|Pr[\succ]-\frac{1}{n!}| \le \frac{\binom{k + n}{n}}{(k + 1)^n} - \frac{1}{n!}$.
\end{restated}

\begin{proof}

To prove this, we need only consider the minimum and maximum probabilities of a specific total order, $\succ$, separately.

Following the bucket analysis in the proof of Theorem~\ref{theorem-ordering-equality-2}, the ordering of any pair of commands is fixed if they fall in different $\Dnet$ buckets.

Thus, a tight lower bound can be given by the probability of all $n$ commands falling into different buckets in that exact order, which corresponds to exactly combinations.

$$Pr[\succ] \ge \frac{\binom{k + 1}{n}}{(k + 1)^n}$$

On the other hand, the adversary has full control over the ordering of commands falling into the same buckets.
So, a tight upper bound can be given by the probability of all $n$ commands falling into buckets in that exact order, but allowing for the same buckets, which corresponds to exactly combinations with repetitions or multisets.

$$Pr[\succ] \le \frac{\multiset{k + 1}{n}}{(k + 1)^n} = \frac{\binom{k + n}{n}}{(k + 1)^n}  $$

Among those two, the upper bound is at least as far as the lower bound from the uniformly random probability of $\frac{1}{n!}$ for a total order.
To see this, let $x = k + 1$, so that $n!\binom{k + n}{n} = x\prod_{j=1}^{n-1}(x + j)$ and $n!\binom{k + 1}{n} = x\prod_{j=1}^{n-1}(x - j)$.
When the two products are expanded in powers of $x$ and added, the terms with an odd number of $j$ factors cancel, and the remaining terms are non-negative, so
$$\prod_{j=1}^{n-1}(x + j) + \prod_{j=1}^{n-1}(x - j) \ge 2x^{n-1}.$$
Hence, the two bounds sum to at least $\frac{2}{n!}$.
And thus we have $\epsilon(n) = \frac{\binom{k + n}{n}}{(k + 1)^n} - \frac{1}{n!}$.

\end{proof}

\subsection{Uniform continuous distribution}

We used the uniform continuous distribution across $[0, \Dnoise]$ in the results shown in Section~\ref{s:eval}.
Here, we give further theoretical analysis of its properties.

\begin{theorem}[$\epsilon(2)$ for uniform continuous distribution]
    \label{theorem-continuous-eps-2}
    The uniform continuous distribution achieves $\epsilon(2) = \frac{\Dnet}{\Dnoise} - \frac{1}{2}(\frac{\Dnet}{\Dnoise})^2$.
\end{theorem}

\begin{proof}
    For the case of two commands invoked at the same time, the adversary can decide their order if the two noise samples differ by less than $\Dnet$.
    Thus, we do an integral to compute this probability (assuming one command goes first):
    \begin{align*}
     & \epsilon(2) \\
     = & \frac{1}{\Dnoise^2} (\int_{0}^{\Dnoise - \Dnet} \int_{x}^{x + \Dnet} 1 dydx \\
    & + \int_{\Dnoise - \Dnet}^{\Dnoise} \int_{x}^{\Dnoise} 1 dydx)
    \\
    = & \frac{\Dnet}{\Dnoise} - \frac{1}{2}(\frac{\Dnet}{\Dnoise})^2
    \end{align*}
\end{proof}

The following bound is only tight under the assumption of an adaptive adversary that can change its picks after seeing the SRO's output.
Because a non-adaptive adversary is strictly weaker, the following bound stays sound but is not necessarily tight.

\begin{theorem}[$\epsilon(n)$ for uniform continuous distribution]
    \label{theorem-continuous-eps-n}
    For $n \ge 2$, the uniform continuous distribution achieves $\epsilon(n) = \frac{1}{n!}((1 + \alpha)^n - 1 - n\alpha^n)$,
    where $\alpha = \frac{\Dnet}{\Dnoise}$.
\end{theorem}
\begin{proof}
    It suffices to prove a tight lower bound and a tight upper bound on any $Pr[\succ]$, the probability of a specific total order.

    For the lower bound,
    consider the time interval $[T + \Dnet, T + \Dnoise]$.
    No matter what the adversary picks,
    any command sent at time $T$ has a probability of $1 - \alpha$ being assigned a modified timestamp uniformly sampled from the interval.
    This is because the assigned timestamp, $t$, is picked from $[T, T + \Dnet)$ and the additional noise is sampled uniformly from $[0, \Dnoise]$.
    Thus, the probability of all $n$ modified timestamps being contained in this interval is $(1 - \alpha)^n$.
    When this occurs, their ordering is uniformly random across all $n!$ permutations, giving a lower bound of $\frac{1}{n!}(1 - \alpha)^n$.

    This lower bound is tight by assigning the first command in $\succ$ an assigned timestamp that approaches $T + \Dnet$ and the last command in $\succ$ an assigned timestamp of $T$.
    In this case, all commands' modified timestamps must be included in this interval for $\succ$ to be the output order.

    For the tight upper bound of $\frac{1}{n!}((1 + \alpha)^n - n\alpha^n)$, we prove it by induction on $n$.

    The base case, $n = 1$, holds trivially.

    For the inductive case, consider how the adversary can pick assigned timestamps to maximize the probability of $\succ$.

    For the first command in $\succ$, $i_1$, because the adversary would like its modified timestamp to be as small as possible, the optimal strategy is assigning an assigned timestamp of exactly $T$ and thus a modified timestamp uniformly sampled from $[T, T + \Dnoise]$.

    There are two possible outcomes for $t'_1$: $t'_1 \le T + \Dnet$ and $t'_1 > T + \Dnet$.

    In the first case, the adversary can assign an assigned timestamp to every other command of at least $t'_1$ to ensure $i_1$ is the first command in the permutation.
    But the adversary must still make the other commands follow $\succ$.
    The optimal strategy is to achieve the tight upper bound of $n-1$ commands with a different $\alpha' = \alpha - \frac{t'_1 - T}{\Dnoise}$, which gives the following integral over the probability density function:

    \begin{align*}
          & \int_{0}^{\alpha}\frac{1}{(n - 1)!}((1 + (\alpha - x))^{n-1} - (n - 1)(\alpha - x)^{n - 1}) dx \\
        = & \int_{0}^{\alpha}\frac{1}{(n - 1)!}((1 + y)^{n-1} - (n - 1)y^{n - 1}) dy \\
        = & \frac{1}{n!}((1 + \alpha)^n - 1 - (n - 1)\alpha^n)
    \end{align*}

    In the second case, where $t'_1 > T + \Dnet$, the optimal strategy is to delay the latter commands as much as possible.
    Thus, every other command gets an assigned timestamp approaching $T + \Dnet$, and their modified timestamps are all sampled uniformly from $[T + \Dnet, T + \Dnet + \Dnoise]$.
    To have $\succ$ being the output order, all other commands must be in the interval $[t'_1, T + \Dnet + \Dnoise]$, which has a probability of $(1 + \alpha - \frac{t'_1 - T}{\Dnoise})^{n-1}$.
    Furthermore, those commands need to satisfy $\succ$.
    Because all those timestamps are sampled uniformly from the same interval, all of their orderings are equally likely by symmetry, and we have the following integral:

    \begin{align*}
         & \int_{\alpha}^{1} \frac{1}{(n - 1)!}(1 + \alpha - x)^{n-1}dx \\
        = & \int_{\alpha}^{1} \frac{1}{(n - 1)!}y^{n-1}dy \\
        = & \frac{1}{n!}(1 - \alpha^n)
    \end{align*}

    Adding the probabilities of the two cases together, we have:

    \begin{align*}
        Pr[\succ] & \le \frac{1}{n!}((1 + \alpha)^n - 1 - (n - 1)\alpha^n) + \frac{1}{n!}(1 - \alpha^n) \\ & = \frac{1}{n!}((1 + \alpha)^n - n\alpha^n)
    \end{align*}

    This bound is tight following our construction.

    For $n \ge 2$, the upper bound is the one farther from $\frac{1}{n!}$, because $(1 + \alpha)^n + (1 - \alpha)^n - 2 \ge n(n - 1)\alpha^2 \ge n\alpha^n$, which gives the stated $\epsilon(n)$.

\end{proof}

\subsection{Discussion}

The discrete distribution focuses on minimizing $\epsilon(2)$ with a bounded maximum latency against a Byzantine adversary. If the system concerns average latency or a weighted average, or if there is a finer-grained model of adversarial power, the optimal distribution would differ.

Furthermore, while the discrete distribution achieves the optimal bound for $\epsilon(2)$, we have not discovered the optimal bound for $\epsilon(n)$.

Using a uniform continuous distribution would not give the optimal $\epsilon(2)$ bound.
However, it might provide a better fairness/latency trade-off in practice, as the adversary may not be as powerful as assumed.
This is why we used the continuous distribution in our experiments.

\end{appendices}

\fi
\end{document}